\documentclass[a4paper,UKenglish,cleveref,autoref,thm-restate]{lipics-v2021}
\pdfoutput=1 
\hideLIPIcs        
\nolinenumbers     

\title{Complexity of Evaluating GQL Queries}
\titlerunning{Complexity of Evaluating GQL Queries}

\author{Diego Figueira}{CNRS, LaBRI, Univ. Bordeaux, France}{diego.figueira@labri.fr}{https://orcid.org/0000-0003-4715-5096}{}
\author{Anthony W. Lin}{University of Kaiserslautern-Landau and MPI-SWS,
Germany}{awlin@mpi-sws.org}{https://orcid.org/0000-0003-0114-2257}{}
\author{Liat Peterfreund}{Hebrew University of Jerusalem, Israel}{liat.peterfreund@mail.huji.ac.il}{https://orcid.org/0000-0002-4788-0944}{}
\authorrunning{Figueira, Lin and Peterfreund}
\Copyright{Diego Figueira, Anthony W. Lin and Liat Peterfreund}


\EventEditors{To~be~assigned}
\EventNoEds{0}
\EventLongTitle{28th International Conference on Database Theory (ICDT~2025)}
\EventShortTitle{ICDT~2025}
\EventAcronym{ICDT}
\EventYear{2025}
\EventDate{March 25--28, 2025}
\EventLocation{Barcelona, Spain}
\EventLogo{}
\SeriesVolume{TBD}
\ArticleNo{XX}

\usepackage{soul}        
\usepackage{tikz}
\usetikzlibrary{arrows.meta,positioning,fit}
\usepackage{mdframed}
\usepackage{mathtools}
\usepackage{bbding}
\usepackage{algorithm2e}
\usepackage[bibliography=common]{apxproof}
\usepackage{dutchcal}
\usepackage{xspace}
\usepackage{marginnote}
\usepackage{stmaryrd}
\usepackage{multirow}
\usepackage{enumerate}

\usepackage[xcolor, hyperref, cleveref, notion, quotation, electronic]{knowledge}
\usepackage{mathcommand}
\usepackage{amsmath} 
\usepackage{stmaryrd}

\knowledgeconfigure{quotation, protect quotation={tikzcd}}
\knowledgeconfigure{diagnose line=true, diagnose bar=true}

\definecolor{Dark Ruby Red}{HTML}{5d1416}
\definecolor{Dark Gamboge}{HTML}{be7c00}
\definecolor{Dark Blue Sapphire}{HTML}{053641}

\IfKnowledgePaperModeTF{
}{
    \knowledgestyle{intro notion}{color={Dark Ruby Red}, emphasize}
    \knowledgestyle{notion}{color={Dark Blue Sapphire}}
    \hypersetup{
        colorlinks=true,
        breaklinks=true,
        linkcolor={Dark Blue Sapphire}, 
        citecolor={Dark Blue Sapphire}, 
        filecolor={Dark Blue Sapphire}, 
        urlcolor={Dark Blue Sapphire},
    }
    \IfKnowledgeElectronicModeTF{
    }{
        \knowledgeconfigure{anchor point color={Dark Ruby Red}, anchor point shape=corner}
        \knowledgestyle{intro unknown}{color={Dark Gamboge}, emphasize}
        \knowledgestyle{intro unknown cont}{color={Dark Gamboge}, emphasize}
        \knowledgestyle{kl unknown}{color={Dark Gamboge}}
        \knowledgestyle{kl unknown cont}{color={Dark Gamboge}}
    }
}
\knowledge{notion}
| notion@notice
| definition@notice

\knowledge{typewriter, url={https://ctan.org/pkg/knowledge}}
| knowledge

\knowledge{typewriter, url={https://github.com/remimorvan/knowledge-clustering}}
| knowledge-clustering

\knowledge{notion}
| multiset
| multisets

\knowledge{notion}
| size@mset

\knowledge{notion}
 | linear set
 | linear sets

\knowledge{notion}
 | semilinear set
 | semilinear sets
 | semilinear

\knowledge{notion}
| relation name
| relation names
| relation
| relations

\knowledge{notion}
| function name
| function names
| function
| functions

\knowledge{notion}
| constant name
| constant names
| constant
| constants

\knowledge{notion}
| vocabulary

\knowledge{notion}
| relational vocabulary
| relational vocabularies

\knowledge{notion}
| $\sigma$-structure
| $\sigma$-structures
| $\sigma$-structure over
| structure
| structures
| relational structure
| relational structures
| relational $\vocab$-structure
| \(\sigma\)-{structure}
| \(\sigma \)-structures

\knowledge{notion}
| domain

\knowledge{notion}
| active domain

\knowledge{notion}
| $\embed \sigma \istruc$-structure
| $\embed \sigma \istruc$-structures
| embedded structure
| embedded structures
 | $\istruc$-embedded structures
 | $\istruc$-embedded structure

\knowledge{notion}
| term
| terms

\knowledge{notion}
| variable
| variables

\knowledge{notion}
| formula
| formulas

\knowledge{notion}
| $\bar x$-valuation
| valuation
| valuations

\knowledge{notion}
| parameter
| parameters

\knowledge{notion}
| working
| working variable
| working variables

\knowledge{notion}
| aggregate function
| aggregate functions
| aggregate functions over
| aggregate function over

\knowledge{notion}
| restrictors-free
| Restrictors-free

\knowledge{notion}
| evaluation problem
\knowledge{notion}
| evaluation problem@GPC

\knowledge{url={https://en.wikipedia.org/wiki/Subset_sum_problem}}
| subset-sum

\knowledge{notion}
| data complexity

\knowledge{notion}
| log-size assumption

\knowledge{url={https://complexityzoo.net/Complexity_Zoo:D\#dp}}
| DP

\knowledge{notion}
 | Presburger arithmetic
 | Presburger Arithmetic
 | Presburger formula

\knowledge{notion}
 | existential Presburger formula
 | existential Presburger arithmetic
 | existential Presburger Arithmetic

\knowledge{notion}
 | star

\knowledge{notion}
 | Parikh's Theorem

\knowledge{notion}
 | Parikh image

\knowledge{notion}
 | First-order logic
 | first-order logic

\knowledge{notion}
 | Hybrid Second-order Logic
 | Hybrid Second-Order Logic

\knowledge{notion}
 | relation variables
 | relation variable

\knowledge{notion}
 | arity

\knowledge{notion}
 | Existential Hybrid Second-order

\knowledge{notion}
 | Regular data path queries
 | RDPQs
 | RDPQ
 | Regular data path queries over $\istruc$

\knowledge{notion}
 | regular path queries
 | RPQs
 | Regular Path Queries
 | RPQ

\knowledge{notion}
 | embedded signature

\knowledge{notion}
 | extended register automata
 | ERA

\knowledge{notion}
 | guard
 | guards

\knowledge{notion}
 | path
 | paths

\knowledge{notion}
 | facts
 | fact

\knowledge{notion}
 | accepted

 \knowledge{notion}
 | register
 | registers

 \knowledge{notion}
 | read-only register
 | read-only registers

 \knowledge{notion}
 | property graph
 | property graphs
 | Property Graphs

\knowledge{notion}
 | GPC pattern
 | Patterns
 | pattern
 | patterns
 | GPC patterns
 \knowledge{notion}
 | Graph Pattern Calculus
 | GPC
 | GPC queries
 | GPC query

\knowledge{notion}
| positive fragment
| positive
 \knowledge{notion}
 | Boolean GPC queries
 | Boolean GPC query
 
\knowledge{notion}
 | restricted quantifier collapse
 | Restricted Quantifier Collapse
 | RQC

\knowledge{notion}
| 1NF relation
| 1NF

\knowledge{notion}
| 1NF fragment

\knowledge{notion}
 | Linear Arithmetic

\knowledge{notion}
 | Linear Real Arithmetic
 | LRA
 
\knowledge{notion}
| ordered
| ordered structures
\knowledge{notion}
 | Real Ordered Field
 | ROF

\knowledge{notion}
 | restricted quantifier
\knowledge{notion}
 | unrestricted quantifier
 | unrestricted quantifiers

\knowledge{notion}
 | active domain formula
 | active domain formulas
\knowledge{notion}
 | Graph patterns
 | graph patterns
 | graph pattern

 \knowledge{notion}
 | restrictor
 | restrictors
\knowledge{notion}
 | GQL query
 | GQL queries
 | GQL

\knowledge{notion}
 | graph database
 | graph databases

\knowledge{notion}
 | restrictor-free

 \knowledge{notion}
 | Conjunctive Regular Path Queries
 | CRPQs

\knowledge{notion}
 | Regular queries
 | Regular Queries

\knowledge{notion}
 | Nested Positive 2RPQ

\knowledge{notion}
 | Extended CRPQ
 | ECRPQ
 | ECRPQs

\knowledge{notion}
 | automatic relations
 | automatic relation

\knowledge{notion}
 | data graphs

\knowledge{notion}
 | data graph

\knowledge{notion}
 | REM

 \knowledge{notion}
 | GQL queries extended to an infinite structure $\istruc $
 | GQL query $\gpcq $ extended with

\knowledge{notion}
 | simple path semantics

\knowledge{notion}
 | trail semantics

\knowledge{notion}
 | subsume
 | subsumes
 | subsumed

 \knowledge{notion}
 | Path patterns
 | path pattern
 | path patterns
 
 \knowledge{notion}
 | output pattern
 | output patterns 
 
\knowledge{notion}
| free variables
| Free variables

\knowledge{notion}
 | descriptors
 | descriptor

\knowledge{notion}
 | condition
 | conditions

\knowledge{notion}
 | path-binding
 | path-bindings

\knowledge{notion}
 | linear queries
 | Linear queries
 | linear query
 | GQL linear query

\knowledge{notion}
 | clauses
 | clause
 | GPC clauses
 | GPC clause
 | GQL clause

 \knowledge{notion}
 | Linear Integer Arithmetic
 | LIA

\knowledge{notion}
 | infinite structure

\knowledge{notion}
 | $\istruc$-embedded finite relational structure

\knowledge{notion}
 | definable

\knowledge{notion}
 | active-domain quantifier

\knowledge{notion}
 | active-domain formula
 | active-domain $\FO[\istruc]$-formula
 | active-domain $\FOHTC[istruc]$-formula
 | active-domain $\HSO[\istruc]$- (resp. $\HESO [\istruc]$-)formula
 | active-domain $\FOHTC[\istruc]$-formula
 | active-domain $\HSO[\istruc]$-formula
 | active-domain $\HESO[\istruc]$-formula

\knowledge{notion}
 | Hybrid Transitive Closure Logic

\knowledge{notion}
 | Second-order Logic
 | Second-order logic

\knowledge{notion}
 | Existential Second-order

\knowledge{notion}
 | GQL queries with a $\istruc $-data type
 | GQL query $\gpcq $ with a $\istruc $-data type
\knowledge{notion}
 | schema
\knowledge{notion}
 | Atomic conditions
 | atomic conditions

\knowledge{notion}
 | Composite conditions

 \knowledge{notion}
 | evaluation problem@GQL
 | evaluation@GQL

 \knowledge{notion}
 | basic fragment of GQL
 | basic fragment
 | basic

\knowledge{notion}
 | node IDs
 | node ID

\knowledge{notion}
 | directed edge IDs
 | edge IDs
 | undirected edge IDs
 | (directed) edge IDs

\knowledge{notion}
 | labels

\knowledge{notion}
 | properties

\knowledge{notion}
 | keys

\knowledge{notion}
 | nulls

\knowledge{notion}
 | graph 

\knowledge{notion}
 | Path Patterns

\knowledge{notion}
 | tuples

\knowledge{notion}
 | Odd-Index problem

\knowledge{text={i.e.}, italic}
  | ie

\knowledge{text={s.t.}, italic}
  | st

\knowledge{text={e.g.}, italic}
  | eg

\knowledge{text={vs.}, italic}
  | vs

\knowledge{text={w.r.t.}, italic}
  | wrt

\knowledge{text={a.k.a.}, italic}
  | aka

\knowledge{text={w.l.o.g.}, italic}
  | wlog

\knowledge{text={cf.}, italic}
  | cf

\knowledge{text={iff}, italic}
  | iff

\knowledge{wrap=\textsf}
  | NL

\knowledge{wrap=\textsf}
  | PH

\knowledge{notion, text={paraNL}, wrap=\textsf}
  | para-NL
  | paraNL

\knowledge{notion, wrap=\textsf}
  | FPT
  | fixed-parameter tractable

\knowledge{text={ExpSpace}, wrap=\textsf}
  | EXPSPACE
  | ExpSpace

\knowledge{text={2ExpSpace}, wrap=\textsf}
  | 2EXPSPACE
  | 2ExpSpace

\knowledge{text={PSpace}, wrap=\textsf}
  | PSpace
  | PSPACE

\knowledge{text={FP}, wrap=\textsf}
  | FP

\knowledge{text={NP}, wrap=\textsf}
  | NP

  \knowledge{text={co-NP}, wrap=\textsf}
  | coNP

  \knowledge{text={P}, wrap=\textsf}
  | P

  \knowledge{text={\ensuremath{\# \textsf{P}}}, wrap=\textsf}
  | shP

  \knowledge{text={\ensuremath{\textsf{P}^{\# \textsf{P}}}}, wrap=\textsf}
  | PshP
  
\knowledge{text={\ensuremath{\textsf{FP}^{\# \textsf{P}}}}, wrap=\textsf}
  | FPshP

\knowledge{text={\ensuremath{\Sigma^p_2}}, wrap=\textsf}
  | SigmaP2

\knowledge{notion, text={\ensuremath{\text{P}^{\text{NP[log]}}}}, wrap=\textsf}
  | PNPLOG

\knowledge{text={\ensuremath{\Pi^p_2}}, wrap=\textsf}
  | PiP2
  | Pi2
  | pi2
  | Pitwo
  | PiTwo
  | pitwo

\knowledge{url={https://en.wikipedia.org/wiki/Savitch\%27s_theorem}}
  | Savitch's Theorem

\theoremstyle{plain}

\theoremstyle{definition}
	\newtheorem*{example*}{Example}          

\newtheoremrep{theorem}{Theorem}[section]
\newtheoremrep{proposition}{Proposition}[section]

\newcommand{\struct}{\mathbf{S}}

\newcommand{\reals}{\mathbb R}

\renewcommand{\phi}{\varphi}
\renewcommand{\leq}{\leqslant}

\renewcommand{\emptyset}{\varnothing}
\newcommand{\set}[1]{\{#1\}}
\newrobustcmd{\tup}[1]{\langle #1 \rangle}
\newrobustcmd{\defeq}{\mathrel{\hat{=}}}

\newrobustcmd{\N}{\mathbb{N}}
\newrobustcmd{\Z}{\mathbb{N}}
\newrobustcmd{\Q}{\mathbb{Q}}
\newcommand{\dcup}{\mathop{\dot\cup}} 
\newrobustcmd\pset[1]{\wp(#1)} 
\knowledgenewrobustcmd{\mset}[1]{\cmdkl{\{\!\!\{}%
  #1%
  \cmdkl{\}\!\!\}}%
}

\knowledgenewrobustcmd{\Rel}{\cmdkl{\textsf{Rel}}} 
\knowledgenewrobustcmd{\Const}{\cmdkl{\textsf{Const}}} %
\knowledgenewrobustcmd{\Func}{\cmdkl{\textsf{Func}}} %
\knowledgenewrobustcmd{\arity}{\cmdkl{\textit{ar}}} 
\newcommand{\tinysep}{\hspace{-.265em}}
\knowledgenewrobustcmd{\interp}[2]{\cmdkl{[\tinysep[}#1\cmdkl{]\tinysep]^{#2}}}
\knowledgenewrobustcmd{\sem}[3]{\cmdkl{\llbracket}#1\cmdkl{\rrbracket^{#3}_{#2}}}%
\knowledgenewrobustcmd{\semRel}[3]{\cmdkl{[\tinysep[}#1\cmdkl{]\tinysep]^{#3}_{#2}}}%
\knowledgenewrobustcmd{\semPA}[1]{\cmdkl{[\tinysep[}#1\cmdkl{]\tinysep]}}%
\knowledgenewrobustcmd{\modelsem}{\mathrel{\cmdkl{\models}}}%
\newcommand{\logicOp}[1]{\textup{\textsf{#1}}}
\knowledgenewrobustcmd{\AggTC}{\cmdkl{\logicOp{AggTC}}}
\knowledgenewrobustcmd{\TC}{\cmdkl{\logicOp{TC}}}
\knowledgenewrobustcmd{\TCemb}{\cmdkl{\logicOp{TC}}} 
\knowledgenewrobustcmd{\Aggr}{\cmdkl{\logicOp{Agg}}} 

\renewcommand{\epsilon}{\varepsilon}
 
\newrobustcmd\labelwithproof[1]{%
\label{#1}%
\ifproofappendix%
\marginnote{\footnotesize{%
  \textnormal{First stated in page~\pageref{#1}.}%
}}
\else%
\marginnote{\footnotesize{%
  \ifarxiv%
    \textnormal{See the proof of \Cref{#1} in page~\pageref{proof-#1}.}%
  \else%
  \fi%
}}%
\fi%
}

\newrobustcmd\introinrestatable[1]{%
\ifproofappendix%
\kl{#1}%
\else%
\intro{#1}%
\fi%
}

\newrobustcmd\introinrestatableopt[1]{%
\ifproofappendix%
\kl[#1]{#1}%
\else%
\intro[#1]{#1}%
\fi%
}

\newrobustcmd\recall[1]{
  \proofappendixtrue%
    #1*
  \proofappendixfalse%
}

\newrobustcmd\Crefapdx[1]{%
  \ifarxiv%
    \Cref{apdx:#1}%
  \else%
    the "full version"%
  \fi%
}

\usepackage[backgroundcolor=orange!20, textcolor={Dark Ruby Red}, textsize=tiny]{todonotes}
\definecolor{green}{RGB}{0,120,0}
\definecolor{hlyellow}{RGB}{250, 250, 190}
\definecolor{diegoeditcolor}{RGB}{210,210,255}
\definecolor{liateditcolor}{RGB}{244,210,255}
\definecolor{santieditcolor}{RGB}{255,255,180}
\definecolor{michaeleditcolor}{RGB}{200,165,90}
\newcommand{\sidediego}[1]{\todo[backgroundcolor=diegoeditcolor, size=\tiny]{{\fontsize{5}{7}\selectfont {\bf D:} #1}}}

\definecolor{light-gray}{gray}{0.9}

\newcommand{\sideliat}[1]{\todo[backgroundcolor=liateditcolor, size=\tiny]{{\fontsize{5}{7}\selectfont {\bf L:} #1}}}

\definecolor{light-gray}{gray}{0.9}

\newcommand{\OMIT}[1]{}

\newrobustcmd{\wrote}{\color{wrote}\scriptsize\text{wrote}}
\newrobustcmd{\advised}{\color{advised}\scriptsize\text{advised}}

\newcommand{\ie}{\textit{i.e.}}

\newcommand{\eg}{\textit{e.g.}}

\knowledgenewrobustcmd{\A}{\mathbb{A}} 
\knowledgenewrobustcmd{\Aext}{\cmdkl{\mathbb{A}^\pm}} 
\knowledgenewrobustcmd\vertex[1]{\cmdkl{V}(#1)}
\knowledgenewrobustcmd\edges[1]{\cmdkl{E}(#1)}

\knowledgenewrobustcmd\qvar{\footnotesize\bullet} 

\knowledgenewrobustcmd{\pathl}{\cmdkl{\mathbf{P}_{\!l}}} 

\knowledgenewrobustcmd\subaut[3]{#1\cmdkl{[#2,#3]}}

\knowledgenewrobustcmd\bagmap{\cmdkl{\mathbf{v}}}
\knowledgenewrobustcmd\tagmap{\cmdkl{\mathbf{t}}}
\knowledgenewrobustcmd\tagmappath[1]{\cmdkl{\mathbf{t}[#1]}}
\newrobustcmd\tagmappathprime[1]{%
  \withkl{\kl[\tagmappath]}{%
    \cmdkl{\mathbf{t}'[#1]}%
  }%
}

\knowledgenewrobustcmd{\atom}[1]{\,\xrightarrow{\smash{#1}}\,}
\knowledgenewrobustcmd{\coatom}[1]{\,\xleftarrow{\smash{#1}}\,}
\knowledgenewrobustcmd{\atoms}[1]{\cmdkl{\textnormal{Atoms}}(#1)}
\knowledgenewrobustcmd{\contained}{\mathrel{\cmdkl{\subseteqq}}}
\newrobustcmd{\strcontained}{
  \mathrel{\withkl{\kl[\contained]}{\cmdkl{%
    \subsetneqq
  }}}
}
\knowledgenewrobustcmd{\semequiv}{\mathrel{\cmdkl{\LaTeXequiv}}}
\knowledgenewrobustcmd{\nbatoms}[1]{\cmdkl{\|}#1\cmdkl{\|}}
\knowledgenewrobustcmd{\vars}{\cmdkl{\textit{vars}}} 
\newrobustcmd{\collapse}{\approx}
\knowledgenewrobustcmd{\Exp}{\cmdkl{\textnormal{Exp}}} 

\knowledgenewrobustcmd{\UCtwoRPQ}{\cmdkl{\textnormal{UC2RPQ}}}
\newrobustcmd{\CtwoRPQ}{%
  \withkl{\kl[\UCtwoRPQ]}{\cmdkl{%
    \textnormal{C2RPQ}
  }}
}

\knowledgenewrobustcmd{\UCRPQSRE}{\ensuremath{\cmdkl{\textup{UCRPQ}(\textup{SRE})}}}
\newrobustcmd{\CRPQSRE}{%
  \withkl{\kl[\UCRPQSRE]}{\cmdkl{%
    \textup{CRPQ}(\textup{SRE})
  }}
}

\knowledgenewrobustcmd{\homto}{\mathrel{\cmdkl{\xrightarrow{\smash{\textit{\tiny hom}}}}}}
\newcommand{\xrightarrowdbl}[2][]{%
  \xrightarrow[#1]{#2}\hspace{-.8em}\xrightarrow{}
}
\knowledgenewrobustcmd\surj{\mathrel{\cmdkl{\xrightarrowdbl{\smash{\textit{\tiny hom}}}}}}
\knowledgenewrobustcmd{\fun}{f}

\knowledgenewrobustcmd{\class}{\mathcal{C}}
\knowledgenewrobustcmd{\Tw}[1][k]{\cmdkl{\mathcal{T\!w}_{#1\!}}}

\knowledgenewrobustcmd{\Refin}[1][]{\cmdkl{\textnormal{Ref}^{\smash{#1}}}}
\knowledgenewrobustcmd{\MUA}[2]{\cmdkl{\ensuremath{\textnormal{App}_{#2}(#1)}}}
\knowledgenewrobustcmd{\MUAHom}[2]{\cmdkl{\ensuremath{\textnormal{App}_{#2}^{\smash{\star}}(#1)}}}
\knowledgenewrobustcmd{\MUAHomBounded}[3]{\cmdkl{\ensuremath{\textnormal{App}_{#2}^{\smash{\star,#3}}(#1)}}}
\knowledgenewrobustcmd{\type}{\cmdkl{\textnormal{type}}}
\knowledgenewrobustcmd{\Qapp}{\cmdkl{\ensuremath{\textnormal{App}_{\Tw}^{\textup{zip}}(\gamma)}}}
\knowledgenewrobustcmd{\contract}[1]{\cmdkl{[}#1\cmdkl{]}}

\newcommand{\gpckw}[1]{\texttt{#1}}

\knowledgenewrobustcmd{\NodeSet}{\cmdkl{\mathcal{N}}}
\knowledgenewrobustcmd{\EdgeSet}{\cmdkl{\mathcal{E}}}
\knowledgenewrobustcmd{\LabelSet}{\mathcal{L}}
\knowledgenewrobustcmd{\PropertySet}{\mathcal{P}}
\knowledgenewrobustcmd{\KeySet}{\mathcal{K}}
\newrobustcmd{\ConstSet}{\mathcal{C}}

\newcommand{\NodeRel}{\mathsf{Node}}

\newcommand{\VarSet}{\Vars}

\newcommand{\labelf}{\mathsf{lbl}}
\newcommand{\srcf}{\mathsf{src}}
\newcommand{\tgtf}{\mathsf{tgt}}
\newcommand{\propf}{\mathsf{prop}}

\newcommand{\df}{:=}

\knowledgenewrobustcmd{\gpcppat}{\cmdkl{\pi}}
\knowledgenewrobustcmd{\gpcgpat}{\cmdkl{\Pi}}
\knowledgenewrobustcmd{\gpcq}{\cmdkl{\mathcal{Q}}}
\knowledgenewrobustcmd{\fv}[1]{\cmdkl{\mathsf{fv}}#1}
\knowledgenewrobustcmd{\return}{\cmdkl{\Omega}}

\newcommand{\gpcterm}{\chi}

\newcommand{\restrict}{\rho}

\knowledgenewrobustcmd{\dom}{\cmdkl{\textit{dom}}}%
\knowledgenewrobustcmd{\adom}{\cmdkl{\textit{adom}}}%

\knowledgenewrobustcmd{\semgpc}[1]{\cmdkl{\llbracket}#1\cmdkl{\rrbracket}}

\newcommand{\join}{\bowtie}
\newcommand{\comp}{\sim}

\newcommand{\gdb}{{G}}

\newcommand{\src}[1]{\mathsf{src}(#1)}
\newcommand{\len}[1]{\mathsf{len}(#1)}
\newcommand{\tgt}[1]{\mathsf{tgt}(#1)}
\newcommand{\sch}[1]{\mathsf{schema}(#1)}

\knowledgenewrobustcmd{\domf}{\cmdkl{\textit{dom}}}%
\knowledgenewrobustcmd{\schgpc}{\cmdkl{\textit{sch}}}%

\knowledgenewrobustcmd{\withAr}[2]{#1^{\cmdkl{(}#2\cmdkl{)}}}%
\knowledgenewrobustcmd{\embed}[2]{\cmdkl{\langle#1,#2\rangle}}%
\knowledgenewrobustcmd{\Vars}{\cmdkl{\textsf{Vars}}} 
\knowledgenewrobustcmd{\FO}{\cmdkl{\textup{FO}}}%
\knowledgenewrobustcmd{\rexists}{\cmdkl{\exists^{\adom}}}%
\knowledgenewrobustcmd{\FOTC}{\cmdkl{\textup{FO[TC]}}}%
\knowledgenewrobustcmd{\FOpTC}{\cmdkl{\textup{FO[TC}^+\textup{]}}}%
\knowledgenewrobustcmd{\FOHTC}{\cmdkl{\textup{FO[HTC]}}}%
\knowledgenewrobustcmd{\FOESO}{\cmdkl{\textup{FO[ESO]}}}
\knowledgenewrobustcmd{\ESO}{\cmdkl{\textup{ESO}}}%
\knowledgenewrobustcmd{\HSO}{\cmdkl{\textup{HSO}}}%
\knowledgenewrobustcmd{\HESO}{\cmdkl{\textup{HESO}}}%
\knowledgenewrobustcmd{\SO}{\cmdkl{\textup{SO}}}%
\knowledgenewrobustcmd{\RVars}{\cmdkl{\textsf{Var}^{\textsf{rel}}}} 
\knowledgenewrobustcmd{\existsRadom}{\cmdkl{\exists^{\adom}}}%
\knowledgenewrobustcmd{\EHSO}{\cmdkl{\textup{EHSO}}}%
\newcommand{\istruc}{\frak M}
\knowledgenewrobustcmd{\RDPQ}{\cmdkl{\textup{RDPQ}}}
\knowledgenewrobustcmd{\RQM}{\cmdkl{\textup{RQM}}}

\knowledgenewrobustcmd{\graphvoc}{\cmdkl{\sigma_{\mathsf{graph}}}}
\knowledgenewrobustcmd{\graphvoctyp}{\cmdkl{\sigma^{\mathsf{t}}_{\mathsf{graph}}}}

\knowledgenewrobustcmd{\restrictorTrail}{\cmdkl{\textsf{trail}}}
\knowledgenewrobustcmd{\restrictorAcyclic}{\cmdkl{\textsf{acyclic}}}
\knowledgenewrobustcmd{\restrictorShortest}{\cmdkl{\textsf{shortest}}}
\knowledgenewrobustcmd{\restrictorAll}{\cmdkl{\textsf{all}}}
\knowledgenewrobustcmd{\restrictorAny}{\cmdkl{\textsf{any}}}

\knowledgenewrobustcmd{\strLA}{\cmdkl{\istruc_{\N,+,<}}}%
\knowledgenewrobustcmd{\strRLA}{\cmdkl{\istruc_{\reals, +, \leq}}}%
\knowledgenewrobustcmd{\strROF}{\cmdkl{\istruc_{\reals, +, \times, \leq}}}%

\NewCommandCopy{\proofqedsymbol}{\qedsymbol}
\newcommand{\exampleqedsymbol}{{$\triangle$}}
\AtBeginEnvironment{example}{%
  \pushQED{\qed}\renewcommand{\qedsymbol}{\exampleqedsymbol}%
}
\AtEndEnvironment{example}{\popQED\endexample}

\newcommand{\partitle}[1]{\smallskip \textbf{#1.}}

\newcommand{\vocab}{\sigma}
\newcommand{\domain}{D}

\newcommand{\TCzero}{\text{TC}^0}
\newcommand{\ACzero}{\text{AC}^0}
\newcommand{\NL}{\text{NL}}

\begin{document}
	\maketitle
	
	\begin{abstract}
		GQL has recently emerged as the standard query language over graph databases
(particularly, the property graph model). Indeed, this is analogous to the role
of SQL for relational databases. Unlike SQL, however, fundamental problems
regarding GQL are hitherto still unsolved, most notably the complexity of query
evaluation. In this paper we provide a complete solution to this problem. In particular, we show that the data
complexity of GQL is $\text{P}^{\text{NP}[\log]}$-complete in general, and is $\text{NL}$-complete, when the so-called ``restrictors'' are disallowed. Using techniques from
embedded finite model theory, we show that this is true, even when the queries use data from infinite concrete domains (for example the domain of real numbers where arithmetic is allowed in the query). In proving these results, we establish and exploit tight connections between GQL and query languages over relational databases, especially the extension of relational calculus with transitive closure operators, 
and a fragment of second-order logic.

	\end{abstract}
	
	
	\section{Introduction}\label{sec:intro}
Knowledge graphs based on the property-graph paradigm have gained an increasing amount of popularity in the last
decades, with numerous application domains, including social networks, knowledge graphs, semantic web, biological and ecological databases, and more.
On the practical side, there has been rapid growth in the development of
commercial graph database systems, including 
systems such as Neo4j, Oracle, IBM, 
TigerGraph, and JanusGraph,
among others. 
This wealth of available systems has also motivated the recent and still ongoing \emph{standardization} effort of graph data models (in particular, \emph{property graphs}) and query languages (in particular, GQL), which has been undertaken by several working groups consisting of research leaders from academia and industry (see,~\eg,~\cite{LDBC:TR:TR-2021-01,gql}).

GQL \cite{GQLStandards} was developed to fill in the role of SQL (i.e., as a
yardstick query language) for property graphs. Unlike SQL whose complexity
landscape has been mapped for decades, the rigorous study of "GQL" is only just
taking shape. A first informal description of the language appeared in
\cite{DBLP:conf/sigmod/DeutschFGHLLLMM22}, followed by a formal treatment of its
pattern-matching layer in~\cite{gpc-pods}. The relational-algebra layer was
analysed soon after in~\cite{ICDT23}, and an expressiveness study subsequently
distilled a streamlined theoretical core~\cite{Gheerbrant2025GQL}. Consequently,
even the most basic issue --- that of \emph{data complexity} (query evaluation
with only the database as input) --- remains largely open. 
Hitherto, we know only that the complexity of enumerating query's outputs is in "PSpace" \cite{gpc-pods}, and we have an NP-hard lower bound easily derivable from~\cite{MartensNP23}. 
In this paper, we identify the precise data complexity of "GQL" for the first 
time.


\textbf{GQL on property graphs in a nutshell.} 
\begin{figure*}
  \begin{center}
    \includegraphics[width=.8\textwidth]{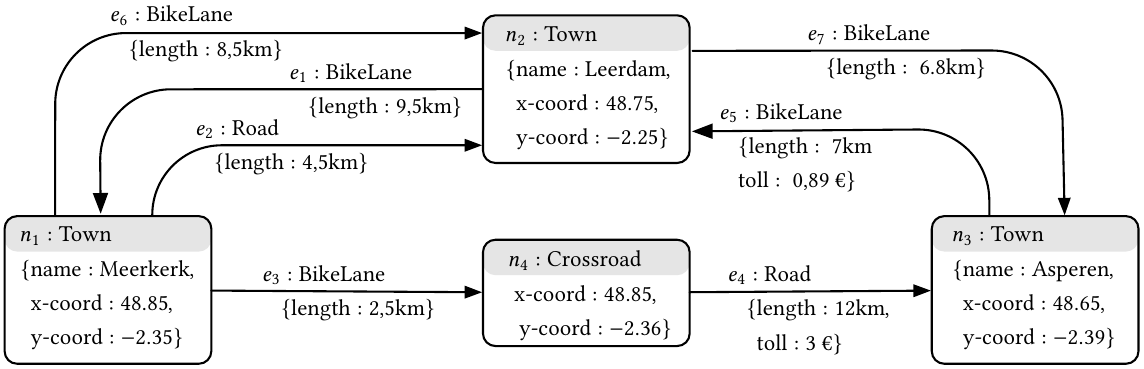}
    \caption{A "property graph" $G_{\label{fig:lanes-example}\mathtt{bike}}$ of bike lanes and road connections between towns and crossroads.}
  \end{center}
\end{figure*}
\emph{A property graph} \cite{pablo-survey} is a  graph whose nodes and edges are labeled, often to represent their type or role in the graph, and are enriched with additional attributes called properties. These properties are stored in a key-value format, enabling a flexible and detailed representation of both entities and their relationships. 
    Figure~\ref{fig:lanes-example} depicts a property graph $G_{\mathtt{bike}}$ of bike lanes and roads between towns and crossroads. It has two node types indicated by their labels---Town and Crossroad, and two edge types, namely Bike-lane and Road. Towns have properties like names, and coordinates specifying their exact location, Bike-lanes and Roads have length and possibly toll. 

GQL queries are essentially SQL queries over relations of nodes, edges, and properties, extracted using patterns that capture both data and graph topology (i.e., paths).
\begin{example}
To find all pairs $(x,y)$ of towns accessible by a bike-lane we  apply the pattern
\begin{equation*}\label{eq:bikelane}
\left(    (x\text{:Town}) 
\xrightarrow{\text{:BikeLane}}
(y\text{:Town}) \right)_{x,y}
\end{equation*}
on  $G_{\mathtt{bike}}$.
The output is a binary relation (with attributes $x,y$)
consisting of, among others, the pair $(n_2,n_1)$. 
To find pairs of places accessible by one or more bike-lanes, we can incorporate \emph{unbounded repetition} of edges:
\begin{equation}\label{eq:bikelane_rep}
R(x,y) \df \left(    (x\text{:Town}) 
\xrightarrow{\text{:BikeLane}}^{1..\infty}
(y\text{:Town}) \right)_{x,y}
\end{equation} 
Here, the output is the transitive closure of the previous output. 
We can refine matched patterns by appending a Boolean filter, e.g.,
returning every town
$y$ reachable \emph{by bicycle} from Leerdam with one or more bike-lanes:
\begin{equation*}
	\left(    (x\text{:Town}) 
	\xrightarrow{\text{:BikeLane}}^{1..\infty}
	(y) \langle x.\text{name} = \text{``Leerdam''} \rangle  \right)_{y}
\end{equation*} 
The condition $\langle x.\text{name} = \text{``Leerdam''} \rangle$ filters the matches to paths starting at Leerdam.
To prevent re-using the same bike-lane or revisiting the same place we prefix the pattern with an appropriate restrictor: we use \gpckw{simple} to forbid repeating a node, or \gpckw{trail} to forbid re-using an edge.
{For instance, suppose we want to organize a day trip from Meekerk to Leerdam using at least two bike lanes without passing through the same town twice. In this case, we would use the pattern
\begin{equation*} 
 \gpckw{simple}
	\left(    (x) 
	\xrightarrow{\text{:BikeLane}}^{2..\infty}
	(y) \langle x.\text{name} = \text{``Meekerk''} \wedge y.\text{name} = \text{``Leerdam''} \rangle  \right)
\end{equation*}
Here, the \gpckw{simple} restrictor enforces the non-repetition of nodes.
}

In GQL, SQL-like-queries are applied to the relations produced by pattern matching over the graph. To query a round-trip, out by bicycle on bike lanes and back by car on roads, we combine the pattern in (\ref{eq:bikelane_rep}), namely $R(x,y)$, with the  pattern 
\begin{equation*}
S(x',y')\df\,	\left(    (x'\text{:Town}) 
	\xrightarrow{\text{:Road}}^{1..\infty}
	(z'\text{:Town}) \right)_{x',y'}
\end{equation*}
 via
$
\pi_{x,y} \left(\sigma_{y=x' \land x=z'}\left(R(x,y) \times S(x',z')\right)\right)
$
that joins the two relations on equalities 
$y=x' \wedge x=z'$.
\end{example}

In this paper, we study the \emph{query evaluation problem} of GQL: given a "GQL" query $Q$, a property graph $G$, and a candidate tuple $\bar{t}$, determine whether 
$ \bar{t}$ belongs to the result of evaluating $Q$ over $ G $.
To study the complexity of the problem, we use \emph{data complexity} \cite{Vardi82}, which is the standard complexity measure for query evaluation problems. That is, since we typically deal with queries of ``small size'' and ``large'' databases, we consider the query component $Q$ as \emph{fixed} (i.e., not part of the input) and only measure the computational complexity in terms of the data (i.e. $G$). 
This regime reflects typical reasoning scenarios, where a fixed query or rule is evaluated against a large and evolving knowledge base, e.g., checking whether a large social network satisfies  ``six-degrees-of-separation''.

\partitle{Contributions} Our main result is summarized here:
\begin{theorem}
    The data complexity of "GQL" queries is "PNPLOG"-complete. The data complexity of "GQL" queries \emph{without restrictors} improves to "NL"-complete.
    \label{th:main}
\end{theorem}
That is, "GQL" query evaluation is representative of problems that can be solved in polynomial time with logarithmically many calls to "NP" oracles. Note that "PNPLOG" lies in the second-level of polynomial hierarchy (which is in "PSpace"), and subsumes the entire boolean hierarchy (which is a boolean closure of "NP"-complete problems). The class "PNPLOG" includes some natural problems related to optimization (e.g. \cite{Wagner87,SV00}) and logical reasoning (e.g. \cite{GMT09,Gottlob95}).
In the absence of path restrictors, the "NP"-hardness from \cite{MartensNP23} on querying simple paths or trails of even length  no longer applies. Indeed, we show that the complexity in this case drops to "NL"-complete.

Our proof technique for the upper bound complexity is via a relational embedding of "GQL" to extensions of first-order logic over relational structures, namely, either with transitive closure operator, or with existential second-order definable relational views. 
This relational viewpoint of property graphs yields several benefits, including applicability of techniques from \emph{embedded finite model theory} (see \cite[Chapter 13]{libkinbook} and \cite{benediktsurvey}) to derive the same data complexity of "GQL" in the presence of numeric data and arithmetic constraints (e.g. constraints from Tarski's theory of real field). In particular, the framework essentially allows \emph{unrestricted quantification} over an infinite domain such as the set of real numbers, which allows us to define interesting queries over spatial databases, e.g., there is a path from $x$ to $y$ that uses edges that lie on a straight-line (see \cite[Chapter13]{libkinbook}). We summarize this in the following corollary:
\begin{corollary}
    Let $\struct$ be a structure satisfying the property of Restricted Quantifier Collapse ("RQC"), in particular Linear Real Arithmetic ("LRA"), Linear Integer Arithmetic ("LIA"), or Real-ordered Fields ("ROF"). Then, the data complexity of "GQL" extended with unrestricted quantification over $\struct$ is "PNPLOG"-complete and, without restrictors, "NL"-complete.
\end{corollary}


\textbf{Organization.} 
Section~\ref{sec:prelim} recalls preliminaries on property graphs, "GQL", and general logic. 
Section~\ref{sec:gql} presents our main complexity results for GQL. 
In Section~\ref{sec:domspecific}, we extend the analysis to incorporate data values and domain-specific operators. Section \ref{sec:overview} highlights the benefits of analyzing graph query languages from a logical perspective.
We conclude in Section~\ref{sec:conc}. 
Some details are deferred to Supplementary Materials.

	\section{Preliminaries: "Property Graphs", "GQL" and Logic}\label{sec:prelim}
	
We follow the definitions of~\cite{Gheerbrant2025GQL} and recall the property graph data model, and the formal syntax and semantics of GQL. 
We also introduce some classical definitions from logic.


\partitle{Property Graphs}
\AP
 We assume pairwise disjoint countable infinite sets of ""nodes"" $\intro*\NodeSet$, ""(directed) edges"" $\intro*\EdgeSet$, 
""labels"" $\intro*\LabelSet$, ""properties"" $\intro*\PropertySet$ 
and ""keys"" $\intro*\KeySet$. A ""property graph"" $G$ is a tuple
$
	(N^G, E^G, 
	\srcf^G, \tgtf^G, \propf^G, \labelf^G )
$
where $N^G\subset \NodeSet, E^G\subset  \EdgeSet$,
are finite sets of nodes and edges identifiers, respectively, and 
$\srcf^G,\tgtf^G  :E^G\rightarrow N^G$ are functions that specify the source and target, respectively, of an edge,
$\propf^G: (E^G\cup N^G) \times \KeySet \rightarrow \PropertySet$ is a partial function that maps a given edge or node and a key to the corresponding property value, if defined, and
$\labelf^G\subset (E^G\cup N^G )\times {\LabelSet}$ specifies the labels attached to nodes and edges. We omit the superscrupt $G$ when it is clear from the context.

\subsection{GQL Syntax}
We follow the core "GQL" definition~\cite{Gheerbrant2025GQL}, and  set an infinite set of variables $\Vars$.
\AP
"Path patterns" $\intro*\gpcppat$ are defined via mutual recursion along with their \AP""free variables"" $\intro*\fv{(\gpcppat)}$.
\AP
Pattern matching is the key component of GQL. 
""Path patterns"" are defined recursively as follows:
\[
\begin{array}{lllll}
\gpcppat &\df& (x) &&\text{node pattern}\\[2pt]
&\mid& \overset{x}{\rightarrow} && \text{edge pattern}\\[2pt]
&\mid& \overset{x}{\leftarrow} &&\text{edge pattern}\\[2pt]
&\mid& \gpcppat_1 \,\gpcppat_2 &&\text{concatenation}\\[2pt]
&\mid& \gpcppat_1 + \gpcppat_2 &\text{if }\fv(\gpcppat_1)=\fv(\gpcppat_2)&\text{disjunction}\\[2pt]
&\mid& \gpcppat^{\,n..m} &&\text{repetition}\\[2pt]
&\mid& \gpcppat \langle \theta \rangle &&\text{filtering}
\end{array}
\]
where
\begin{itemize}
	\item $x \in \Vars$ and 
	$0 \le n \le m \le \infty$;
	\item variables $x$ in node patterns $(x)$,  and edge patterns
	$\overset{x}{\rightarrow}$ and $\overset{x}{\leftarrow}$ are optional,
	\item $\gpcppat\langle\theta \rangle$ is a conditional pattern, where
	conditions  $\theta$ are given by 
	$\theta \df \ x.k=x'.k' \mid
	\ell(x) \footnote{This condition checks whether $x$ is labeled with $\ell$. For convenience, we often use the shorthand $(x{:}\ell)$, which incorporates the label directly into the node pattern.}
    \mid 
	\theta \vee \theta
	\mid
	\theta \wedge \theta
	\mid \neg \theta$
	where $x,x'\in\Vars$ , $\ell\in\LabelSet$, and 
	$k,k'\in\KeySet$;
\end{itemize}

\noindent The free variables of path patterns are defined as follows:
\[\begin{array}{rl}
\fv{\left((x)\right)} = \fv({\overset{x}{\rightarrow}}) =  \fv({
		\overset{x}{\leftarrow}}) &\df \ \{x\}\\  \fv({\gpcppat_1 + \gpcppat_2} )
	& \df\  \fv({\gpcppat_1})
	\\ \fv({ \gpcppat_1\, \gpcppat_2 }) & \df\ \fv({\gpcppat_1})\cup\fv({ \gpcppat_2 })
	\\ \fv({ \gpcppat^{n..m} }) & \df \ \emptyset
	\\ \fv({\gpcppat\langle\theta \rangle}  )& \df\ \fv({\gpcppat})
\end{array} \]
\noindent
\AP We then define 
\emph{""output patterns""} as $\intro* \bar{\restrict} \gpcppat_{\return}$  where 
\begin{itemize}
	\item $\bar \restrict  \df [\gpckw{shortest}][\gpckw{simple}|\gpckw{trail}]
	$ is sequence of  \AP""restrictors"" such that  parts within squared brackets
	are optional, and 
	\item $\return$ is a (possibly empty) finite sequence of elements of the form $x$ and $x.k$ where $x\in \Vars$ and $k\in\KeySet$. 
\end{itemize}

Notice that while~\cite{Gheerbrant2025GQL} omitted  "restrictors" from the language, we include  them and give them a precise semantics that is consistent with the standard~\cite{ICDT23}, extending the core fragment and aligning the formal model more closely with the practical language.

\AP Finally, we associate every output pattern $\bar\restrict\gpcppat_{\return}$ with a relation symbol $R_{\bar\restrict\gpcppat_{\return}}$.  
""GQL queries"" \(\gpcq\) are then defined as the closure of these symbols under the standard operators of relational algebra.\footnote{
Notice that 
while \cite{Gheerbrant2025GQL} uses a linear form of relational algebra, we use the more convenient standard relational algebra as the equivalence already established in that work.
} 
Formally,
\[
\gpcq \df R_{\bar\restrict\gpcppat_{\return} }\mid \gpcq \cup \gpcq \mid \gpcq \times \gpcq \mid \gpcq \setminus \gpcq \mid \pi_{\bar x} \gpcq \mid \sigma_{\theta} \gpcq
\]
for every output pattern $\bar\restrict \gpcppat_{\return}$
where $\bar x $ is a vector of $\Vars$ and $\theta \df x = y \mid \theta \vee \theta \mid \theta \wedge \theta \mid \neg \theta$ with $x,y\in \Vars$.
The \(\schgpc({\gpcq})\) of a "GQL" query \(\gpcq\) is defined inductively. For the base case, 
$
\schgpc({R_{\bar\restrict\,\gpcppat_{\return}}})\df \return.
$
For composite queries the schema is defined in the standard way (see, e.g., \cite{AbiteboulHullVianu1995}).

\subsection{GQL Semantics}
\label{app:sem-GQL}

The semantics of output patterns is defined in terms of the semantics of path patterns. We therefore begin by defining the latter.

A \emph{path} \(p\) in a property graph \(G\)         is an alternating sequence of nodes and edges that start and end with a node, that is,
\[
n_0\, e_0\, n_1\, e_1\, \cdots\, e_{k-1}\, n_k
\]                 
where \(n_0, \ldots, n_k \in N\), and \(e_0, \ldots, e_{k-1} \in E \). Each edge \(e_i\) must connect the adjacent nodes, meaning either 
(1) $\srcf(e_i) = n_i${ and } $\tgtf(e_i) = n_{i+1}$ {or} (2) $\tgtf(e_i) = n_i$ { and } $\srcf(e_i) = n_{i+1}$.

The semantics 
$\semgpc{\gpcppat}^G$ 
of a path pattern $\gpcppat$ over a property graph \(G\) is defined as a set of pairs \((p, \mu)\), where:
\begin{itemize}
    \item \(p\) is a path in \(G\), and
    \item \(\mu : \VarSet \to N \cup E\) is a partial assignment that is defined on  \(\operatorname{dom}(\mu) \df \fv(\gpcppat)\).
\end{itemize}

We define $\src{p}\df n_0,\,\tgt{p} \df n_k$, and $\len{p}\df k$.
For two path $p_1, p_2$, we write $p_1 \parallel p_2$ to indicate that $\tgt{p_1} = \src{p_2}$.
In this case, $p_1 p_2$ is the concatenation of the paths. 
\AP
For partial mappings $\mu_1,\mu_2 :\VarSet \rightarrow N\cup E$, we use  
 $\mu_1\comp \mu_2$ to denote that for every $x\in\domf(\mu_1)\cap \domf(\mu_2) $, it holds that $\mu_1(x)=\mu_2(x)$ and define $(\mu_1 \join \mu_2)(x) \df\mu_1(x)$ if $x \in \domf(\mu_1)$, and 
$
\mu_2(x)$ otherwise.
We represent mappings $\mu$ as sets of elements $x\mapsto o$, where $\mu(x)=o$. (If $\dom(\mu)=\emptyset$ then we represent $\mu$ by $\emptyset$.)

\AP The semantics $\semgpc{\gpcppat}^G$ of "path patterns" 
on a "graph" $G$ is defined in Figure~\ref{fig:sem}.

\newcommand{\semfig}{%
\[
\begin{array}{r@{\hskip 3pt} l}
\multicolumn{2}{l}{\textbf{Path Patterns.}}\\
      \semgpc{(x)}^G \df& \left\{\left(({n}),\{x\mapsto n\} \right) \,\middle| \, 
          n\in N 
    \right\} \\
    \semgpc{\overset{x}{\rightarrow} }^G 
    \df& \left\{ (({n_1, e , n_2}), \{x\mapsto e\} ) \,\middle| \,
         e\in E,\,
          \srcf{(e)}=n_1,\, \tgtf{(e)}=n_2
    \right\}\\
    \semgpc{ \overset{x}{\leftarrow} }^G 
    \df& \left\{ (({n_1, e , n_2}), \{x\mapsto e\}  ) \,\middle| \,
         e\in E,\,
          \srcf{(e)}=n_2,\, \tgtf{(e)}=n_1
    \right\}
\\
 \semgpc{ \pi_1 \pi_2}^G \df& \left\{(p_1  p_2, \mu_1\join \mu_2) \,\middle| \begin{array}{l}
 (p_1,\mu_1)\in \semgpc{\pi_1}^G,\,
  (p_2,\mu_2)\in \semgpc{\pi_2}^G \\
 p_1 \parallel p_2,\,
 \mu_1\sim \mu_2,\,
 \end{array}\right\}  \\
    \semgpc{ \pi_1 +\pi_2}^G \df&    \semgpc{ \pi_1}^G \cup \semgpc{\pi_2}^G \\
     \semgpc{\gpcppat\langle\theta \rangle}^G \df& \left\{ (p,\mu) \,\middle|\,   (p,\mu)\in \semgpc{\gpcppat}^G,\, 
     \semgpc{\theta}^G_{\mu} = \top
     \right\}\\
      \semgpc{ \pi^{n..m}}^G \df& \cup_{i=n}^m \semgpc{\pi}^G_{i} \text{ where } \\
      \semgpc{\pi}_0^G \df& 
      \left\{(({n}),\emptyset)\, \middle|\, 
  n\in N \right\} \text{ and for } k>0:\\
\semgpc{\pi}^G_k \df &
\left\{ (p_1 \cdots p_k,\emptyset)\, \middle|\,
\begin{array}{@{}l@{}}
  (p_1,\mu_1) \in \semgpc{\pi}^G,\ldots,\,
 (p_k,\mu_k) \in \semgpc{\pi}^G
\end{array}
\right\}
\end{array} 
\]
\[\begin{array}{r@{\hskip 1pt}l}
\multicolumn{2}{l}{\textbf{Conditions.}}\\
\semgpc{\ell(x)}^G_{\mu } \df& \top  \text{ if } \labelf(\mu(x)) = \ell \\
\semgpc{x.k = x'.k'}^G_{\mu} \df &  
    \top \text{ if } \propf(\mu(x),k) =
  \propf(\mu(x'),k')
\\
\semgpc{\theta_1 \vee \theta_2}^G_{\mu} \df  &\semgpc{\theta_1 }^G_{\mu}\vee \semgpc{\theta_2}^G_{\mu} \\
\semgpc{\theta_1 \wedge \theta_2}^G_{\mu } \df   & \semgpc{\theta_1 }^G_{\mu}\wedge \semgpc{\theta_2}^G_{\mu}\\
\semgpc{\neg \theta}^G_{\mu} \df &  \neg \semgpc{ \theta}^G_{\mu}
\end{array}
\]
}

\begin{figure}[t]
    \centering
    \begin{minipage}{\linewidth}
        \begin{mdframed}
            \semfig
        \end{mdframed}
    \end{minipage}
    \caption{
        \label{fig:sem}
        Semantics of GQL's "Path Patterns" and "Conditions"}
\end{figure}
For the semantics of "output patterns"
$\restrict \gpcppat_{\return}$, we first define the intermediate semantics of path patterns preceded by restrictors, iterating through possible cases.
For $\restrict \in \{\gpckw{simple}, \gpckw{trail}\}$, we define
$
   \semgpc{\restrict \gpcppat}^G
\df \left\{ 
\mu \, \middle|\, 
(p,\mu)\in \semgpc{\gpcppat}^G,\rho(p) 
\right\}
$ where
$\restrict(p) = \top$ if and only if $p$ is a simple path or a trail, respectively.
Next, for $\tilde{\gpcppat} $ of the form $[\gpckw{simple} |\gpckw{trail}]  \gpcppat$ (with the "restrictors" being optional) we define 
\[
 \semgpc{\gpckw{shortest} \,\tilde{\gpcppat}}^G
\df \left\{ 
(p,\mu) \, \middle|\, 
(p,\mu)\in \semgpc{\tilde{\gpcppat}}^G ,\len{p} = m_p
\right\}
\]
where
$m_p$ is the minimal length among paths that match $\tilde{\gpcppat}$ with the same endpoints as $p$, or, more formally, 
\[
  m_p = \min \left\{ \len{p'}\,\middle|\,
            (p',\mu )\in \semgpc{\tilde{\gpcppat}}^G, \,
         \src{p'}=\src{p},\,
         \tgt{p'}=\tgt{p}
    \right\}
.\]

Let $\return$ be a (possibly empty) finite sequence of elements of the form $x$ and $x.k$ where $x\in \Vars$ and $k\in\KeySet$.
We say that a mapping $\mu$ is \emph{compatible} with $\Omega$ if for each  $x \in \Omega$, $\mu$ is defined on $x$, and for each $x.k \in \Omega$, $\mu$ is defined on $x$ and $\propf(\mu(x),k)$ is defined.
In this case, we define 
$\mu_{\return}: \return \rightarrow N \cup E \cup \PropertySet $ as the projection of
$\mu$ on $\return$:   

\[\mu_{\return} (\omega) \df
\begin{cases} 
\mu(x) & \text{if } \omega  =  x\in \Vars \\
\propf(\mu(x),k)  & \text{if }  \omega = x.k\, 
\end{cases}\]
Finally, we define $\semgpc{\bar \restrict { \gpcppat}_{\return}}^G
\df \left\{ 
\mu_{\return}\,\middle|\, 
(p,\mu)\in \semgpc{\bar \restrict{ \gpcppat}}^G
\right\}
$
Note that each semantics instance of output pattern forms a \emph{relation}: a set of partial mappings that share the same domain. Once this is fixed, the semantics of an entire "GQL" query is obtained exactly as in standard relational algebra (see, e.g., \cite{AbiteboulHullVianu1995}).
\subsection{First and Second-Order Logic}
We assume basic familiarity with mathematical logic (e.g.~\cite{libkinbook}). In the sequel, we deal only with finite relational structures.

\textbf{Relational structures.}
\AP
We fix disjoint countably infinite sets of ""relation names"" \(\Rel\),
""constant names"" \(\Const\), and ""variables"" \(\Vars\).
Every relation name \(R \in \Rel \) is associated with a positive integer, namely its ""arity"" \(\arity(R)\). We write $R^{(\arity({R}))}$ to make the arity explicit.
A ""{vocabulary}"" is a finite set \(\sigma\subseteq\Rel\).
\AP
A ""relation"" is a finite set of ""tuples"" in $\Const^{\arity(R)}$ .
A ""\(\sigma\)-{structure}"" \(\struct\) maps every
\(R^{(k)}\in\sigma\) to a finite relation
\(\interp{R}{\struct}\subseteq\Const^{k}\).
Its \emph{active domain} is
\(\adom(\struct)=\{c\in\Const: c\text{ occurs in }\interp{R}{\struct}\text{ for every }R\in\sigma\}\).
Property graphs can naturally be viewed as 
$\sigma$-structures over the vocabulary consisting of the relation names 
$N, E, \srcf, \tgtf, \propf, \labelf$ and interpreted accordingly. 

\textbf{First-order logic.}
\(\FO\) is defined in the usual way over  "\(\sigma\)-structures".
That is, terms $t$ are defined by $t\df x \mid c$ where $x\in \Vars$, $c\in\Const$ and formulas $\varphi \df
      t = t'
      \;\mid\;
      R(t_1,\dots,t_k)
      \;\mid\;
      \lnot\,\varphi
      \;\mid\;
      \varphi \wedge \varphi
      \;\mid\;
      \exists x\,\varphi $.
We use the standard $\forall, \rightarrow$ as syntactic sugar. 

For a formula $\varphi$,
we write $\varphi(\bar x)$ to explicitly indicate that $\bar x$ are its free variables.
An \AP""$\bar x$-valuation"" is a function $\nu$ mapping "variables" of $\bar x$ to elements of the "active domain" of a "relational structure" $\struct$.
We write \(\sem{\phi(\bar x)}{\nu}{\struct}\) for the truth value $\top$ or $\bot$ of
\(\phi\) under the "valuation" \(\nu:\bar x\to\adom(\struct)\). We denote by \(\sem{\phi(\bar x)}{}{\struct}\) the set of "valuations" $\nu$ that satisfy $\phi$ in $\struct$, that is, "valuations" for which  \(\sem{\phi(\bar x)}{\nu}{\struct}= \top\). 

\textbf{Transitive closure.}
First-order logic extended with the transitive-closure operator is denoted \FOTC. Transitive-closure formulas are of the form
$
\TC_{\bar u,\bar v}\bigl(\phi(\bar u,\bar v,\bar p)\bigr)(\bar x,\bar y)
$,
where $\bar p \bar x \bar y$ are the free variables. Such a formula is true in a structure \(\struct\) under a "valuation" \(\nu\) iff there exists a finite sequence of tuples
$
\bar d_{0}, \bar d_{1}, \dots, \bar d_{n} \;\in\;\adom(\struct)^{|\bar u|}
$
with \(\bar d_{0} = \nu(\bar x)\) and \(\bar d_{n} = \nu(\bar y)\), such that for every \(i < n\),
$\sem{\phi}{\mu_i}{\struct} = \top$, where $\mu_i$ is the "valuation" mapping $\bar u$, $\bar v$ and $\bar p$ to $\bar d_i$,  $\bar d_{i+1}$ and  $\nu(\bar p)$ respectively.

\textbf{Second-order logic.}
\SO\ (second-order logic) extends \FO~by allowing quantification not only over individual elements of the active domain, but also over relation symbols of arbitrary arity. Concretely,  \SO\ formulas may include subformulas of the form
$$
\forall R\,\bigl(\psi(R)\bigr)
\quad\text{or}\quad
\exists R\,\bigl(\psi(R)\bigr),
$$
where $R$ is a relation variable and $\psi(R)$ is an \FO\ formula in which $R$ may occur  as $R(x_1,\dots,x_k)$. 

The fragment \ESO\ consists of those \SO\ formulas in which all second-order quantifiers are existentially quantified. Note that negation in \ESO\ is permitted only within the first-order component and cannot apply to second-order quantifiers. 

\begin{toappendix}

We present the standard definitions of logic, along with their classical complexity bounds.

\textbf{Relational Structures.}
\AP
We set $\intro*\Rel$ and $\intro*\Const$ to be disjoint countably infinite sets of ""relation names"" and ""constant names"", respectively, and assume the existence of an ``arity'' function $\intro*\arity : \Rel $
\AP
We sometimes write $\intro*\withAr{R}{k}$  to denote 
that $\arity(R)=k$.
A ""relational vocabulary"" is a finite set $\sigma \subseteq \Rel$ of relation names. 
\AP
A (finite) ""$\sigma$-structure"" is a function $\struct$ which maps each "relation" $R \in \sigma$ to a (finite) set $\intro*\interp{R}{\struct} \subseteq \Const^{\arity(R)}$.
When $\sigma$ is clear from the context, we call $\struct$ a \reintro{relational structure}.
\AP
The ""active domain"" of $\struct$, namely $\intro*\adom(\struct)$,  is the set
$\set{c \in \Const : c \text{ appears in }\interp{Z}{\struct}, Z \in \sigma}$. 
We say that a structure is ""ordered"" if it is equipped with a designated binary relation that defines a total order on its "active domain". 

\textbf{First-Order Logic.}
We set  $\intro*\Vars$ to be a countably infinite set of ""variables"". 
\AP
Given a $k$-tuple of "variables" $\bar x$, a $\bar x$-""valuation"" is a function $\nu$ mapping "variables" of $\bar x$ to elements of the "active domain" of a "relational structure" $\struct$, usually simply denoted by a $k$-tuple from $\adom(\struct)^k$ (whose $i$-th element is $\nu(\bar x[i])$).
\AP
We define first-order logic $\intro*\FO$  over "relational vocabularies":
\begin{itemize}
    \item $t = t'$ is a "formula" for every $t,t' \in \Const \dcup \Vars$
    \item $\withAr{R}{k}(t_1, \dotsc, t_k)$ is a "formula" for every $t_1, \dotsc, t_k \in \Const \dcup \Vars$
    \item $\lnot \varphi$, $\exists x ~ \varphi$, $\varphi \wedge \psi$ are "formulas" if $\varphi, \psi$ are "formulas".
\end{itemize}
\AP
For any $\bar x$-"valuation" $\nu$ and "$\sigma$-structure" $\struct$, we define the semantics $\intro*\sem{\phi(\bar x)}{\nu}{\struct} \in \set{\top,\bot}$ and $\reintro*\sem{t}{\nu}{\struct} \in \Const$ for any "formula" or "term" whose free variables are in $\bar x$. 
The semantics of $\sem{\phi(\bar x)}{\nu}{\struct}$ and $\sem{t(\bar x)}{\nu}{\struct}$ are as expected (they can be found in \Cref{sec:FO-sem}).

\AP
 \textbf{Transitive closure.}
We extend $\FO$ with transitive closure to obtain  $\intro*\FOTC$ by allowing a new type of "formula" 
    $\intro*\TC_{\bar u,\bar v}(\phi(\bar u \bar v \bar p))(\bar x, \bar y) $ where
 $\phi$ is a formula, and $\bar u, \bar v, \bar x, \bar y$ are of the same dimension $k$.
 We refer to $\bar p$ as ``""parameters""''.
The semantics is such that  $\sem{\TC_{\bar u,\bar v}(\phi(\bar u \bar v \bar p))(\bar x, \bar y)}{\nu}{\struct} = \top$ if there exists a sequence of $k$-tuples of "active domain" elements $\bar d_0, \dotsc, \bar d_n$ such that 
   \begin{itemize}
       \item $\sem{\phi}{\bar u \bar v \bar p \mapsto \bar d_i \bar d_{i+1} \nu(\bar p)}{\struct} = \top$ for every $i$, 
       \item $\bar d_0 = \nu(\bar x)$, and 
       \item$\bar d_n = \nu(\bar y)$.
   \end{itemize} 

   \AP
    We define $\sem{\phi}{}{\struct}$ as the set of all assignments $\nu$ such that $\sem{\phi}{\nu}{\struct} = \top$.
 We refer to the fragment of $\FOTC$ in which every $\TC$ subformula appears under an even number of negations by $\intro*\FOpTC$. 
\begin{proposition}\label{prop:lib}\cite[Propositions~10.22 \& 10.23]{libkinbook}
    In terms of expressive power, on "ordered structures", $\FOTC$ is equivalent to $\FOpTC$, and it "subsumes" all "NL" properties.
\end{proposition}

\textbf{Second-Order Logic.}
\label{sec:SO}
\AP
""Second-order Logic"" ($\intro*\SO$) is another extension of $\FO$, this time with second-order quantifiers over the "active domain". Concretely, we have now another disjoint set $\intro*\RVars$ of ""relation variables"" with arity. We abuse notation and use $\withAr{R}{k}$ and $\arity(R)$ also for "relation variables" $R$.
We add two new kinds of "formulas". 
The first one of the form
$R(t_1,\dotsc, t_{\arity(R)})$ 
where $R$ is a "relation variable" and  every $t_i$ is a "term".
\AP
The second one of the form
$\intro*\exists R ~ \phi(R,\bar S,\bar x),$
where $\phi$ is a "formula", $R$ is a "relation variable", and $\bar S$ and $\bar x$ are tuples of "relation variables" and "variables", respectively.
For the semantics, we now have that "relation variables" assignments are functions from $\RVars$ to \emph{sets} of elements of the "active domain" of $\struct$.
We extend the semantics $\sem{\phi}{\nu,\rho}{\struct}$, where $\nu$ is a "variable" assignment and $\rho$ is a "relation variable" assignment in the expected way.
In particular,
$\sem{\exists R ~ \phi(R,\bar S,\bar x)}{\nu,\rho}{\struct} = \top$ if for some set $A \subseteq (\adom(\struct))^{\arity(R)}$ we have $\sem{\phi(R,\bar S,\bar x)}{\nu,\rho[R \mapsto A]}{\struct} = \top$.
We have $\sem{R(t_1, \dotsc, t_{\arity(R)})}{\nu,\rho}{\struct} = \top$ for $R \in \RVars$ if $(\sem{t_1}{\nu,\rho}{\struct}, \dotsc, \sem{t_k}{\nu,\rho}{\struct}) \in \rho(R)$. $\sem{\phi}{}{\struct}$ is defined analogously.

\AP
The ""Existential Second-order"" fragment $\intro*\ESO$ of $\SO$ is the set of "formulas" built up from first-order formulas via the positive Boolean operators, existential (first-order) quantification, and quantifiers of the form $\exists R_1, \dotsc, R_n ~ \phi$.
 \begin{proposition}
     In terms of expressive power, on ordered structures, $\FOTC$ is included in $\ESO$,
     and $\FOTC$ is equivalent to $\FOpTC$ (and, in turn, to "NL").
 \end{proposition}
The following connection can be easily shown.

\begin{propositionrep}\label{prop:fotcineso}    
    In terms of expressive power, $\FOpTC$ is "subsumed" in $\ESO$.
\end{propositionrep}
This, along with \Cref{prop:lib}, enables us to conclude that on "ordered structures", $\FOTC$ is subsumed in $\ESO$.

\begin{proof}
    it suffices to replace each occurrence of 
    \[[\TC_{\bar u,\bar v}(\phi(\bar u \bar v \bar p))](\bar x, \bar y)\]
    with
    \[
        \exists \withAr{R}{2k} ~ R(\bar x \bar y) \land \textit{lin}(R) \land \forall \bar u \bar v ~ (succ(R,\bar u, \bar v) \rightarrow \phi(\bar u \bar v \bar p))
    \]
    where $\textit{lin}(R)$ expresses that $R$ is a strict linear order,\footnote{That is, it defines a linear order $<$ on the domain $(\adom(\struct))^k$, where $\bar a < \bar b$ if  $\bar a \bar b \in R$.} and $succ(R, \bar u, \bar v)$ expresses that $\bar v$ is the successor of $\bar u$ in the linear order.
\end{proof}

\subsection{Complexity of First and Second-Order Logic}

We are interested in the \AP""data complexity"" of evaluation ("ie", the complexity where the query is considered to be fixed). The 
 ""evaluation problem"" for a logic $\+L$  is the problem of determining, given a formula $\phi(\bar x)$ in $\+L$, a "structure" $\struct$ and a "valuation" $\nu$, 
 whether $\sem{\phi(\bar x)}{\nu}{\struct} = \top$. The following proposition is well-known (e.g. see \cite{libkinbook}).

 \begin{proposition}
    \label{prop:logics-complexities}
    \hfill
    \begin{enumerate}
        \item The "evaluation problem" for $\FOTC$ is "NL"-complete in data complexity.
        \item The "evaluation problem" for $\ESO$ is "NP"-complete in data complexity.
        \item The "evaluation problem" for $\SO$ is "PH"-complete in data complexity.
    \end{enumerate}
 \end{proposition}

\AP
In general, we shall say that a logic $\+L$ ""subsumes"" a query language $\+Q$ (resp.\ a set of properties $\+P$), if for every query $q \in \+Q$ (resp.\ property $p \in \+P$) there is a sentence $\phi \in \+L$ (\ie, a formula with no free variables) which is equivalent, in the sense that for every "relational structure" $\struct$, we have $\sem{\phi}{}{\struct} \neq$ "iff" $q$ yields a non-empty answer on $\struct$ (resp.\ $\struct$ verifies property $p$).

\end{toappendix}

\begin{toappendix}

\subsection{First-order Logic Semantics}
\label{sec:FO-sem}
The full semantics of first-order logic is defined as follows:
\begin{itemize}
    \item $\sem{\exists x ~ \phi(x,\bar y)}{\bar y \mapsto \bar b}{\struct} = \top$ if for some $a \in \adom(\struct)$ we have $\sem{\phi(x,\bar y)}{(\bar y \mapsto \bar b) \cup (x \mapsto a)}{\struct} = \top$,
    \item $\sem{t = t'}{\nu}{\struct} = \top$ if $\sem{t}{\nu}{\struct} = \sem{t'}{\nu}{\struct}$,
    \item $\sem{\withAr{R}{k}(t_1, \dotsc, t_k)}{\nu}{\struct} = \top$ if $(\sem{t_1}{\nu}{\struct}, \dotsc, \sem{t_k}{\nu}{\struct}) \in \interp{R}{\struct}$,
    \item $\sem{c}{\nu}{\struct} = c$ for every $c \in \Const$,
    \item $\sem{x}{\nu}{\struct} = \nu(x)$ for every $x \in \Vars$,
    \item $\sem{\phi \land \psi}{\nu}{\struct} = \top$ if $\sem{\phi}{\nu}{\struct} = \top$ and $\sem{\psi}{\nu}{\struct} = \top$
    \item $\sem{\lnot \phi}{\nu}{\struct} = \top$ if $\sem{\phi}{\nu}{\struct} \neq \top$
\end{itemize}

\end{toappendix}

	\section{Complexity Bounds for "GQL" Queries}\label{sec:gql}
	
We now use classical logics to capture "GQL" with and without restrictors, turning syntactic translations into complexity bounds.


\OMIT{
\begin{mdframed}
\noindent\textbf{ Query Evaluation for GQL}\\[4pt]
\textbf{Input:} A "GQL" query $\gpcq$, a graph $G$, a candidate tuple $\bar t$.\\[3pt]
\textbf{Output:} \textsc{Yes} if  $\bar t \in \semgpc{\gpcq}^{G}$, and \textsc{No} otherwise.
\end{mdframed}

In what follows, we study the \emph{data complexity} of this problem, \ie the complexity of answering it when $\gpcq$ is regarded fixed, and $G$ as input.
}

\subsection{Restrictor-Free GQL Queries and \FOTC}

\AP We define ""restrictor-free"" "GQL queries" as queries in which path patterns are not preceded by "restrictors".
\AP We call a "GQL query" ""restrictor-free"" when no "output pattern" is prefixed by  "restrictors".

\begin{lemma}\label{thm:gql_fotc}
	For every "restrictor-free" "GQL query" $\gpcq$ 
	there is an $\FOTC$ formula $\phi_{\gpcq}(\bar x)$ where $\bar x = \schgpc({\gpcq})$ such that  
	$\semgpc{\gpcq}^{\gdb} =
	\sem{\phi_{\gpcq}(\bar x)}{}{G}
	$
	for every "graph" $G$.
\end{lemma}
\begin{proof}[Proof Sketch]
The proof translates a \emph{restrictor-free} "GQL" query  into an
$\FOTC$ formula whose answer relation coincides with the query
result. The construction involves three main steps: 
\smallskip

\noindent\textbf{(1) Path patterns.}
For every path pattern $\pi$ we build an $\FOTC$ formula
$\phi_\pi(s,t,\bar x)$ whose free variables are
the start node $s$ and end node $t$ of the matched path, and the pattern variables
$\bar x=\fv(\pi)$.  
The translation is recursive:

\begin{itemize}
	\item \emph{Atomic patterns} use predicates
	such as $N(x)$ for node patterns, and combination of $E(e)$, $\srcf(e,s)$, $\tgtf(e,t)$ for edge patterns, depending on their direction.
	\item \emph{Union and concatenation.}  
	$\pi_1+\pi_2$ translates to disjunction, and for
	$\pi_1\pi_2$ we use an existentially quantified concatenation node.
	\item \emph{Filters.}  A condition
	$\pi{\langle\theta\rangle}$ becomes
	$\phi_{\pi}\wedge\phi_\theta$, where each atomic condition
	(\(x{:}\ell\),\; $x.a=y.b$, etc.) is rendered, in turn, into FO.
	\item \emph{Kleene bounds.}
	Bounded repetition $\pi^{k..m}$ unrolls into a finite disjunction of "FO" formula;
	unbounded repetition $\pi^{0..\infty}$ is captured with a
transitive closure operator applied to $\phi_{\pi}$.
\end{itemize}

\smallskip
\noindent\textbf{(2) Output patterns.}
An output pattern $\pi_{\return}$ returns a tuple of
node/edge IDs and property values.
To obtain the output pattern formula, we take $\phi_\pi$ and conjoin, for every attribute in the return list $\Omega$, a first-order equality that identifies that attribute with the corresponding witness variable already introduced by $\phi_\pi$.

\smallskip
\noindent\textbf{(3) Full queries.}
A restrictor-free "GQL" query is just a relational-algebra expression
over the auxiliary relations defined in step (2).  Because every RA
operator is FO-definable, an easy induction 
pushes all definitions inside and yields a single $\FOTC$ sentence
$\phi_{\gpcq}(\bar x)$ whose free variables $\bar x=\schgpc(\gpcq)$.
\end{proof}

\begin{toappendix}
\subsection{Proof of \Cref{thm:gql_fotc}}
\label{app:proof-thm:gql_fotc}

We start with proving the following lemmas:
\begin{lemma}\label{lem:gpcandfotc}
For every "path pattern" $\gpcppat$ there is an $\FOTC$~formula $\phi_{\gpcppat}(s,t,\bar x)$ such that $\bar x = \fv{(\pi)}$, $s,t\not \in \bar x$, and for every "graph" $G$ 
the following hold:
\begin{itemize}
\item 
If $\exists p,\mu: (p,\mu) \in \semgpc{\pi}^{G}$ 
then 
$\sem{\phi_{\gpcppat}(s,t,\bar x)}{st\bar x\mapsto \src{p} \tgt{p} \mu(\bar x) }{G} = \top$.
\item 
If $\sem{\phi_{\gpcppat}(s,t,\bar x)}{st\bar x\mapsto y_1 y_2 \bar z }{G} = \top$ then $\exists p,\mu: (p,\mu) \in \semgpc{\pi}^{G}$ such that $\srcf(p)=y_1$, $\tgtf{(p)}=y_2$ and $\mu(\bar x) = \bar z$.
\end{itemize}
\end{lemma}
\begin{proof}
The proof of this lemma is based on the following recursive translation of $\gpcppat$ to formulas $\phi_{\pi}$.
	
\paragraph{Base cases}

If $\gpcppat \df (x)$ then
\[\phi_{\pi}(s,t,x) \df  N(x)  \wedge s=x \wedge t=x\]
If $\gpcppat \df (\,)$ then
\[\phi_{\pi}(s,t) \df  N(s) \wedge s=t\]
If $\pi \df \overset{x}{\rightarrow}$ then
\[
\phi_{\pi}(s,t,x) \df  E(x) \wedge N(s) \wedge N(t) \wedge \srcf(x,s) \wedge \tgtf(x,t) 
\]
If $\pi \df \overset{}{\rightarrow}$ then
\[
\phi_{\pi}(s,t) \df \exists x:\ \left( E(x) \wedge N(s) \wedge N(t) \wedge \srcf(x,s) \wedge \tgtf(x,t) \right)
\]
If $\pi \df  \overset{x}{\leftarrow}$ then
\[
\phi_{\pi}(s,t,x) \df E(x) \wedge N(s) \wedge N(t) \wedge \srcf(x,t) \wedge \tgtf(x,s) 
\]
If $\pi \df  \overset{}{\leftarrow}$ then
\[
\phi_{\pi}(s,t) \df \exists x:\,\left(E(x) \wedge N(s) \wedge N(t) \wedge \srcf(x,t) \wedge \tgtf(x,s) \right)
\]

\paragraph{Induction step}
If
$\pi = \pi_1 + \pi_2$ then since $\fv({\pi_1})=\fv({\pi_2})$ we set \[\phi_{\pi_1 + \pi_2}(s,t,\bar x) \df \phi_{\pi_1}(s,t,\bar x) \vee \phi_{\pi_2}(s,t,\bar x)\]
If
$\pi = \pi_1 \pi_2$ then 
\[
\phi_{\pi}(s,t, \bar z, \bar x, \bar y) \df \exists n: (\NodeRel(n) \wedge \phi_{\pi_1}(s,n,\bar z,\bar x)\wedge \phi_{\pi_2}(n,t,\bar z,\bar y))
\]
where $\bar z = \fv(\pi_1) \cap \fv(\pi_2) $ are assumed "wlog"~to appear it their stated positions (if it is not the case they can be modified appropriately).

\noindent
Before translating $\pi = \pi'_{\langle \theta \rangle}$, we discuss the condition $\theta$.
We define $\fv(\theta)$ recursively as follows:
\begin{align*}
    \fv (\ell(x))&\df \{ x\}\\
    \fv(x.a=y.b) &\df \{x,y\} \\
    \fv(x.a=c) &\df \{x\} \\
    \fv(\theta_1 \vee \theta_2) &\df \fv(\theta_1)\cup  \fv(\theta_2) \\
    \fv(\theta_1 \wedge \theta_2) &\df \fv(\theta_1)\cup  \fv(\theta_2) \\
    \fv(\neg \theta) &\df  \fv(\theta)
\end{align*}

If $\pi = \pi'_{\langle \theta \rangle}$  then
\[\phi_{\pi}(s,t,\bar x) \df \phi_{\pi'}(s,t,\bar x) \wedge \phi_{\theta}(\fv(\theta))
\]
where $\phi_{\theta}(\fv(\theta))$ is defined as follows:
\begin{itemize}
\item If $\theta = \ell(x)$, then $\phi_{\theta}(x) \;\df\; \labelf(x,\ell)$.
\item If $\theta = x.a = y.b$, then $\phi_{\theta}(x,y) \;\df\; \exists p\,\bigl(
\propf(x,a,p) \wedge \propf(y,b,p)
\bigr).
$
\item If $\theta = \theta_1 \circ \theta_2$ with
$\circ \in \{\vee,\wedge\}$ and $\bar z$ the variables common
to $\fv(\theta_1)$ and $\fv(\theta_2)$, then
$\phi_{\theta}(\bar z,\bar x,\bar y) \;\df\;
\phi_{\theta_1}(\bar z,\bar x) \;\circ\;
\phi_{\theta_2}(\bar z,\bar y).$
\item If $\theta = \neg \theta'$, then
$
\phi_{\theta}(\bar x) \;\df\; \neg \phi_{\theta'}(\bar x).
$
\end{itemize}

Notice that the semantics of $ \pi'_{\langle \theta \rangle}$ is only defined when $\fv(\theta) \subseteq \fv(\pi')$. Therefore, the free variables of $\phi_{\theta}$ are a subset of the free variables of $\phi_{\pi'}$; hence, the above translation is well-defined. 
	
If
$\pi = \pi'^{k..k}$ with $k<\infty$ is translated to 
\[
\phi_{\pi}(s,t) = \exists s_1, \cdots,s_k: \bigwedge_{i=1}^{k-1}\phi'(s_i,s_{i+1}) \wedge \phi'(s_k,t) \wedge s=s_1
\]

The pattern 
$\pi = \pi'^{k..m}$ with $0\le k\le m <\infty$ can be viewed as syntactic sugar to $\pi'^{k}+ \cdots +\pi'^{m} $ which is translated according to the previous item and the $+$ translation.
	
The pattern  $\pi = \pi'^{0..\infty}$  is translated to 
\[
\phi_{\pi}(s,t) \df [\TC_{u,v}(
\phi_{\pi'}(u,v))](s,t)
\]
(Notice that since we are dealing with the "1NF" fragment of "GQL", it holds that  $\schgpc(\pi')= \emptyset$, hence, $\schgpc(\phi_{\pi'}) = \{s,t\}$.)
	
The pattern  $\pi = \pi'^{k..\infty}$  can be viewed as syntactic sugar to $\pi'^{k} \pi'^{0..\infty} $ which is translated according to the translation of previous items and concatenation. 
	
With the above construction at hand, it is straightforward that the condition of the lemma holds.
\end{proof}

We now move to output patterns, and in particular those without restrictors. 
\begin{lemma}\label{lem:outptpattern_fotc}
For every "output pattern" $\gpcppat_{\return}$ there is an $\FOTC$~formula $\phi_{\gpcppat_{\return}}(\bar x)$ such that $\bar x = \fv{(\gpcppat)}$ and for every graph $G$ the following hold:
\begin{itemize}
\item 
If
$ \mu \in \semgpc{\gpcppat_{\return}}^{G}$ then there is $\mu'$ such that  $\sem{\phi_{\gpcppat_{\return}}(\bar x)}{\bar x \mapsto \mu'(\bar x)}{G} = \top$
\item 
If
there exists $\mu$ for which $\sem{\phi_{\gpcppat_{\return}}(\bar x)}{\bar x \mapsto \mu(\bar x)}{G} = \top$ then there is
$\mu' \in \semgpc{\gpcppat_{\return}}^{G}$.
\end{itemize}
\end{lemma}
\begin{proof}
We define $\phi_{\gpcppat_{\return}}(\bar x)$ based on $\phi_{\gpcppat}(y_1,\ldots,y_m)$
whose existence holds from Lemma~\ref{lem:gpcandfotc}.
Let us denote $\Omega \df \omega_1 ,\ldots ,\omega_m$.
Assume wlog that 
$\omega_1, \ldots \omega_i$ are $y_1,\ldots, y_i$, respectively, and that 
$\omega_{i+1}, \ldots \omega_m$ are $y_{i+1}.k_{i+1},\ldots, y_m.k_{m}$, respectively. 
We define 
 $\phi_{\gpcppat_{\return}}(x_1,\ldots, x_m) \df\exists y_1,\ldots,y_m:\,\left(  \phi_{\gpcppat}(y_1,\ldots,y_m) \wedge
\bigwedge_{j={1}}^i x_j=y_j \wedge
 \bigwedge_{j={i+1}}^m \propf(y_{j}.k_{j}) = x_j 
 \right)$.
 Following this construction, it is straightforward to show that both claims hold.
\end{proof}

\begin{lemma}\label{lem:gpcqtofo}
	For every "restrictor-free" "GQL  query" $\gpcq$ there is an \FOTC formula $\phi_{\gpcq}(\bar x)$ such that $\bar x = \sch{\gpcq}$ and for every "graph" $G$ the following holds:
	\[  \semgpc{\gpcq}^{G} =
	\sem{\phi_{\gpcq}(\bar x)}{}{G} 
	\]
\end{lemma}
\begin{proof}
The proof is done by induction on the structure of $\gpcq$, with the base case given by Lemma~\ref{lem:outptpattern_fotc}. The induction step is proved similarly to known results on embedding relational algebra within first order logic (see, \eg, \cite{AbiteboulHullVianu1995}). 
\end{proof}
This lemma allows us to conclude~\Cref{thm:gql_fotc}.
\end{toappendix}

\begin{example}\label{ex:fotc}
   Query~(\ref{eq:bikelane_rep}) from~\Cref{sec:intro} can be expressed as
   $ \phi_1(x,y)\df
   [\TC(\phi(u,v))](x,y)
   $ where 
$
       \phi(u,v) \df \exists e\, \Big(\ E(e) \land \srcf(e,u) \land \tgtf(e,v) \land \labelf(e,\text{``BikeLane''})\Big).
$
Here, the relation names $E$, $\srcf$, $\tgtf$, and $\labelf$ constitute the vocabulary of the graph.
       \end{example}

Combining~\Cref{thm:gql_fotc} with the known fact that the "evaluation problem" for $\FOTC$ is "NL"-complete in data complexity (\cite{libkinbook} gives an "NL"\ upper bound. A matching lower bound via a reduction from graph reachability yields the following tight bound).
\begin{theorem}\label{cor:complexityGQL}
	The data complexity of the "evaluation problem@@GQL" for "restrictor-free" "GQL queries" is in "NL"-complete. 
\end{theorem}

Having settled the complexity of restrictor-free queries, we now turn to queries with restrictors.

\subsection{GQL Queries and \FOESO}

We define \FOESO~as first-order logic extended with \ESO-based first-order views. Specifically, we call an \ESO~formula whose free variables are first-order variables an \emph{ESO FO-view}, and define \FOESO~as the set of first-order formulas that may include such views as sub-formulas.

\begin{lemma}\label{thm:gql_foeso}
    For every "GQL" query $\gpcq$ (with "restrictors") there is an \FOESO~formula $\phi_{\gpcq}(\bar x)$ with $\bar x = \schgpc{(\gpcq)}$ such that $ \semgpc{\gpcq}^{G} = \emptyset$ if and only if $
\sem{\phi_{\gpcq}(\bar x)}{}{\gdb} = \emptyset
    $
for every "graph" $G$.
\end{lemma}
\begin{toappendix}

\subsection{Proof of \Cref{thm:gql_foeso}}
\label{app:thm:GQLtoSO}

We show the following:
\begin{lemma}\label{lem:simple2eso}
For every "path pattern" $\gpcppat$ there is an $\ESO$~formula $\phi^{\gpckw{simple}}_{\gpcppat}(\bar x)$ such that $\bar x = \fv{(\gpcppat)}$ are first order variables, and for every graph $G$ and partial mapping $\mu$ the following are equivalent:
	\begin{itemize}
		\item 
		$\exists p:  ( p,\mu) \in \semgpc{\gpckw{simple}\, \gpcppat}^{G}$ 
		\item 
		$\sem{\phi^{\gpckw{simple}}_{\gpcppat}(\bar x)}{\mu(\bar x)}{G} = \top$
	\end{itemize}
\end{lemma}
\begin{proof}
To incorporate the "restrictor" $\gpckw{simple}$, we need to quantify existentially over a second order variable that restricts the matched paths. 

More formally, we can encode paths $(n_1,e_1,\ldots, n_k,e_k,n_{k+1})$ with a ternary relation $R^{(3)}$ by instantiating it with $R^{(3)}(n_j,e_j,n_{j+1})$ for every $1\le j\le k$.
We encode a ternary edge relation as
$
R(x,e,y)\;,
$
and define the helper predicates
$$
\begin{aligned}
\textsf{edge}(e)   &\;:=\;\exists x\,y\,R(x,e,y),\\[2pt]
\textsf{node}(x)   &\;:=\;\exists e\,y\,\bigl(R(x,e,y)\lor R(y,e,x)\bigr),\\[2pt]
\textsf{source}(x) &\;:=\;\textsf{node}(x)\land\neg\exists e\,y\,R(y,e,x),\\[2pt]
\textsf{target}(x) &\;:=\;\textsf{node}(x)\land\neg\exists e\,y\,R(x,e,y).
\end{aligned}
$$

To ensure that $R^{(3)}$ encodes a simple path, we need to verify that the following holds:
\begin{itemize}
  \item {There is exactly one source}
  $ \phi_1(R) \df
     \exists s\bigl(
        \textsf{source}(s)\land
        \forall s'\,[\,\textsf{source}(s')\!\rightarrow s'=s\,]
     \bigr).
  $
  \item {There is exactly one target}
  $ \phi_2(R) \df
     \exists t\bigl(
        \textsf{target}(t)\land
        \forall t'\,[\,\textsf{target}(t')\!\rightarrow t'=t\,]
     \bigr).
  $
  \item {Every internal node has precisely one in-edge and one out-edge}
  $$
\begin{aligned}
 \phi_3(R) \df \forall n\,\Bigl(&
     \textsf{node}(n)\land\neg\textsf{source}(n)\land\neg\textsf{target}(n)
     \;\rightarrow\\[2pt]
     &\bigl(
        \exists e\,y\,[\,R(n,e,y)\land
           \forall e'\,y'\,(R(n,e',y')\!\rightarrow e'=e\land y'=y)\,]\;\land\\[2pt]
     &\qquad
        \exists e\,y\,[\,R(y,e,n)\land
           \forall e'\,y'\,(R(y',e',n)\!\rightarrow e'=e\land y'=y)\,]
     \bigr)\Bigr).
  \end{aligned}
  $$
  \item \emph{Source node and target node are distinct}
  $ \phi_4(R) \df
     \forall x\,\bigl(\textsf{source}(x)\!\rightarrow\!\neg\textsf{target}(x)\bigr).
  $
\end{itemize}

\noindent
The conjunction of these four sub-formulas is the desired formula. Namely, 
$\psi^{\textsf{simple}}(R) \df \bigwedge_{i=1}^4 \phi_i(R)$.

We now move to define $\phi^{\mathsf{simple}}_{\gpcppat}(\bar x)$:
	\[ \phi^{\mathsf{simple}}_{\gpcppat}(\bar x) \df 
	\exists R:\, \psi^{\mathsf{simple}}(R) \wedge \phi^R_{\gpcppat}(\bar x)
\]
where $\phi^R_{\gpcppat}(\bar x)$ is defined similarly to $\phi_{\gpcppat}(\bar x)$ from the proof of \Cref{lem:gpcandfotc} except that it operates only on simple paths $R$. 
In particular, we change the base cases as follows:

If $\gpcppat \df (x)$ then
\[\phi_{\pi}(s,t,x) \df  \mathsf{node}(x)  \wedge s=x \wedge t=x\]
If $\gpcppat \df (\,)$ then
\[\phi_{\pi}(s,t) \df  \mathsf{node}(s) \wedge s=t\]
If $\pi \df \overset{x}{\rightarrow}$ then
\[
\phi_{\pi}(s,t,x) \df  \mathsf{edge}(x) \wedge \mathsf{node}(s) \wedge \mathsf{node}(t) \wedge \srcf(x,s) \wedge \tgtf(x,t) 
\]
If $\pi \df \overset{}{\rightarrow}$ then
\[
\phi_{\pi}(s,t) \df \exists x:\ \left( \mathsf{edge}(x) \wedge \mathsf{node}(s) \wedge \mathsf{node}(t) \wedge \srcf(x,s) \wedge \tgtf(x,t) \right)
\]
If $\pi \df  \overset{x}{\leftarrow}$ then
\[
\phi_{\pi}(s,t,x) \df \mathsf{edge}(x) \wedge \mathsf{node}(s) \wedge \mathsf{node}(t) \wedge \srcf(x,t) \wedge \tgtf(x,s) 
\]
If $\pi \df  \overset{}{\leftarrow}$ then
\[
\phi_{\pi}(s,t) \df \exists x:\,\left(\mathsf{edge}(x) \wedge \mathsf{node}(s) \wedge \mathsf{node}(t) \wedge \srcf(x,t) \wedge \tgtf(x,s) \right)
\]
The inductive cases remain unchanged.
\end{proof}

We can do this similarly for the "restrictor" $\mathsf{trail}$.
To this end, we drop the ``no repeated nodes'' requirement and instead we forbid re-using an edge.
We define the successor formula on pair of edges 
\[\mathsf{succ}(e,e') \df 
   \exists x\,y\,z\,(R(x,e,y)\land R(y,e',z))
   \]
   and with that at hand, we define
\[
\forall e,e_1,e_2 \;\bigl(\mathsf{succ}(e,e_1)\land\mathsf{succ}(e,e_2)\rightarrow e_1=e_2\bigr)\;\land
\forall e,e_1,e_2 \;\bigl(\mathsf{succ}(e_1,e)\land\mathsf{succ}(e_2,e)\rightarrow e_1=e_2\bigr).
\]

We can hence obtain a similar lemma to~\Cref{lem:simple2eso} that replaces $\gpckw{simple}$ with $\gpckw{trail}$.
\begin{lemma}\label{lem:trail2eso}
For every "path pattern" $\gpcppat$ there is an $\ESO$~formula $\phi^{\gpckw{trail}}_{\gpcppat}(\bar x)$ such that $\bar x = \fv{(\gpcppat)}$ are first order variables, and for every property graph $G$ and partial mapping $\mu$ the following are equivalent:
	\begin{itemize}
		\item 
		$\exists p,\mu:  ( p,\mu) \in \semgpc{\gpckw{trail}\, \gpcppat}^{G}$ 
		\item 
		$\sem{\phi^{\gpckw{trail}}_{\gpcppat}(\bar x)}{\mu(\bar x)}{G} = \top$
	\end{itemize}
\end{lemma}

The previous two lemmas allow us to conclude the following.
\begin{corollary}
    For every "output pattern" $\restrict \gpcppat_{\return}$ there is an ESO formula $\phi_{\restrict \gpcppat_{\return}}(\bar x)$ such that $\bar x = \fv ({\restrict \gpcppat_{\return}})$  are first order variables, and for every graph $G$ the following holds: \[
    \semgpc{\restrict \gpcppat_{\return}}^G  \subseteq \semgpc{\phi_{\restrict \gpcppat_{\return}}(\bar x)}^G
    \]
\end{corollary}
\begin{proof}
    If $\rho = \gpckw{simple} \,| \,\gpckw{trail}$, then the claim follows directly from the previous two lemmas. 
    If $\rho  = \gpckw{shortest} [ \gpckw{simple} | \gpckw{trail}]$ then it follows from the definition of the semantics of $\gpckw{shortest}$ and~\Cref{lem:gpcandfotc}.
\end{proof}

We now move to the proof of ~\Cref{thm:gql_foeso}.
Lemmas~\ref{lem:simple2eso} and~\ref{lem:trail2eso} show that every
{path pattern} equipped with a restrictor
(\texttt{simple} or \texttt{trail}) is definable by an $\ESO$ formula.
Corollary~\ref{cor:restrictor_upper} lifts this to each \emph{output
pattern} $\restrict\,\gpcppat_{\return}$, yielding an
$\ESO$ formula for the auxiliary relation
$R_{\restrict\,\gpcppat_{\return}}$.

A complete "GQL" query $\gpcq$ is obtained by applying the usual
relational-algebra operators to such relations.  Because every RA
operator (projection, selection, join, union, difference) is first-order
definable, we can replace each occurrence of
$R_{\restrict\,\gpcppat_{\return}}$ by its $\ESO$ definition and push
the "FO" operators outward.  This results in an
$\FOESO$ formula
whose evaluation over a graph $G$ is empty exactly when
$\semgpc{\gpcq}^{G}$ is empty. This establishes
Lemma~\ref{thm:gql_foeso}.
\end{toappendix}

\begin{proof}[Proof Sketch]
The proof translates a \emph{restrictor-free} "GQL" query with "restrictors" into a
\FOESO~formula whose  emptiness coincides with the query's emptiness.

\smallskip
\noindent\textbf{(1) Encoding a restricted path.}
We encode a path
\((n_1,e_1,\dots,n_k,e_k,n_{k+1})\)
with a ternary relation
\(R(x,e,y)\)
containing the \(k\) triples
\((n_j,e_j,n_{j+1})\).
For the \texttt{simple} "restrictor" we use the "FO" test
\(\psi^{\mathsf{simple}}(R)\). To define it we use the following auxiliary shortcuts
\[
\begin{array}{rl}
\mathsf{edge}(e)   &\;:=\;\exists x,y\,R(x,e,y),\\
\mathsf{node}(x)   &\;:=\;\exists e,y\,\bigl(R(x,e,y)\lor R(y,e,x)\bigr),\\
\mathsf{source}(x) &\;:=\;\mathsf{node}(x)\land\neg\exists e,y\,R(y,e,x),\\
\mathsf{target}(x) &\;:=\;\mathsf{node}(x)\land\neg\exists e,y\,R(x,e,y).
\end{array}
\]
The conjunction of
the following requirements gives us \(\psi^{\mathsf{simple}}(R)\):
\begin{itemize}
  \item[(i)] {There is exactly one source}
\[
     \exists s\bigl(
        \mathsf{source}(s)\land
        \forall s' (\mathsf{source}(s')\!\rightarrow s'=s)
     \bigr).
\]
  \item[(ii)] {There is exactly one target}
\[
     \exists t\bigl(
        \mathsf{target}(t)\land
        \forall t'(\mathsf{target}(t')\!\rightarrow t'=t)
     \bigr).
\]
  \item[(iii)] {Every internal node (\ie~neither source nor target) has precisely one in-edge and one out-edge}
$ 
 \forall n\,\Bigl(
     \left(\mathsf{node}(n)\land\neg\mathsf{source}(n)\land\neg\mathsf{target}(n)\right)
     \;\rightarrow
     \bigl(
\mathsf{one\_in}(n) \land \mathsf{one\_out}(n)
     \bigr)\Bigr)
$ 
  where $
      \mathsf{one\_in}(n) \df \exists y,e \,(R(y,e,n)\land
           \forall y',e'\,(R(y',e',n) \rightarrow (e'=e\land y'=y)))$ and  $\mathsf{one\_out}(n)$
          is defined symmetrically. 
\end{itemize}
\noindent
For the \texttt{trail} restrictor, we drop condition (iii) and instead require that each edge has at most one successor edge and one predecessor edge. To express this, we define a successor formula over pairs of edges.
\[\mathsf{succ}(e,e') \df 
   \exists x,y,z\,(R(x,e,y)\land R(y,e',z))
   \]
   and with that at hand, we define
\begin{multline*}
\forall e,e_1,e_2 \;\bigl((\mathsf{succ}(e,e_1)\land\mathsf{succ}(e,e_2))\rightarrow e_1=e_2\bigr)\;\land\\
\forall e,e_1,e_2 \;\bigl((\mathsf{succ}(e_1,e)\land\mathsf{succ}(e_2,e))\rightarrow e_1=e_2\bigr).
\end{multline*}
For the \texttt{shortest} restrictor, no special handling is needed, as the emptiness of a pattern is preserved when this "restrictor" is applied.

\smallskip
\noindent\textbf{(2) Changing the atomic patterns.}
For every path pattern \(\gpcppat\) we build, by structural induction,
an $\ESO$ formula
\(
  \phi^{R}_{\gpcppat}(\bar x)
\)
with $\bar x = \fv(\gpcppat)$ that is true iff
\(R\) encodes a restricted path matching \(\gpcppat\).
Only the \emph{base cases} differ from the \FOTC\ construction of~\Cref{thm:gql_fotc}, for example:
\[
\begin{aligned}
\phi_{(x)}(s,t,x) &:= \mathsf{node}(x)\land s=x\land t=x,\\
\phi_{\overset{x}{\rightarrow}}(s,t,x) &:= 
     \mathsf{edge}(x)\land\srcf(x,s)\land\tgtf(x,t),
\end{aligned}
\]
with $\mathsf{node}(x)$ and $\mathsf{edge}(x)$ as reconstructed above from $R$.
A restricted path pattern is therefore represented by
$
\exists R\;\bigl(\psi^{\rho}(R)\land\phi^{R}_{\gpcppat}(\bar x)\bigr),
$
an $\ESO$ sentence with free variables \(\bar x\).

\smallskip
\noindent\textbf{(3) From output patterns to queries.}
Every output pattern \(\bar\restrict\,\gpcppat_{\return}\) introduces a
relation symbol
\(R_{\bar\restrict\,\gpcppat_{\return}}\)
and the $\ESO$ definition obtained in step~2.
A complete "GQL" query is an RA expression over these symbols.
Because each RA operator (projection, join, union, difference) is
FO‐definable, substituting the $\ESO$ definitions and pushing the FO
connectives outward yields an \FOESO\ formula
$\phi_{\gpcq}(\bar x)$ 
satisfying
\(
  \semgpc{\gpcq}^{G}=\emptyset
  \iff
  \sem{\phi_{\gpcq}(\bar x)}{}{G}=\emptyset
\)
for every graph \(G\).
\end{proof}

\begin{example}\label{ex:foeso}
Continuing~\Cref{ex:fotc}, if the pattern is preceded by the "restrictor" $\gpckw{simple}$ then the translation is 
$ \phi_2(x,y)\df
\exists R: \psi^{\gpckw{simple}}(R) \land \tilde{\phi}_1(x,y)
$ where $\tilde{\phi}_1(x,y)$ is obtained from ${\phi}_1(x,y)$ by replacing $E(e)$ with $\mathsf{edge}(e)$ derived from $R$ (as detailed in the proof sketch of~\Cref{thm:gql_foeso}).
\end{example}

The complexity class \AP""PNPLOG"" consists of decision problems solvable in deterministic polynomial time using logarithmically many adaptive queries to an "NP" oracle~\cite{Wagner87,SV00}. Equivalently, it can be defined as the class of problems solvable by first generating polynomially many "NP" oracle queries, performing all of them in parallel (non-adaptively), and then computing the final answer from the oracle’s yes/no responses.

\begin{corollary}\label{cor:restrictor_upper}
    The data complexity of the evaluation problem for "GQL" queries is in  "PNPLOG".
\end{corollary}
\begin{proof}[Proof Sketch]
    Let $\gpcq$ be a "GQL" query and let $\phi_{\gpcq}$ be the \FOESO~formula guaranteed in~\Cref{thm:gql_foeso}.
    Let us denote by $\phi_1, \ldots, \phi_k$ the first-order \ESO~views. Assume that each $\phi_i$ arity $k_i$, and let us denote $n\df |\adom(G)|$.
\begin{enumerate}
  \item For every $i$ and every tuple $\bar a\in \adom(G)^{k_i}$ ask the NP oracles
        whether $\sem{\phi_i}{\bar x \mapsto \bar a}{G} = \top$.  
        The total queries are $\sum_i n^{k_i}=n^{ O(1)}$.
  \item Issue this batch of polynomially many queries \emph{in parallel} to obtain
       the complete satisfying assignments for all of the ${\phi_i}$s.
  \item Evaluate the remaining first-order part of $\phi_{\gpcq}$ deterministically
in $poly(n)$ time.
\end{enumerate}
Notice that the algorithm makes polynomially many non-adaptive NP calls.
\end{proof}

\textbf{Is the bound tight?}
\AP
To investigate this, we introduce the ""Odd-Index problem"". In its original form~\cite[Theorem 5.2]{wagner1987more}, the problem asks: given a list of Boolean formulas \(\phi_1, \ldots, \phi_n\) in 3-CNF, is there an odd index \(i \in [n]\) such that:
\begin{itemize}
	\item all formulas \(\phi_1, \ldots, \phi_i\) are satisfiable, and
	\item all formulas \(\phi_{i+1}, \ldots, \phi_n\) are unsatisfiable?
\end{itemize}
This problem can be generalized  by replacing the 3-CNF satisfiability with any "NP"-hard decision problem. In our context we replace it with \emph{simple-path (trail) even-length reachability} in graphs. Formally, we define the following problem:
\begin{quote}
	Given a list \((G_1, s_1, t_1), \ldots, (G_n, s_n, t_n)\), where each \(G_i\) is a graph and \(s_i, t_i\) are designated nodes, determine whether there exists an \emph{odd index} \(i \in [n]\) such that:
	\begin{itemize}
		\item for all \(j \leq i\), there is a simple path (trail) of even length from \(s_j\) to \(t_j\) in \(G_j\), and
		\item for all \(j > i\), there is no such simple path (trail) of even length.
	\end{itemize}
\end{quote}

We reduce this problem to the evaluation of a  "GQL" query, thereby establishing the following lower bound.
\begin{lemma}\label{lem:restrictors_lower}
	The evaluation problem for "GQL" queries with \emph{restrictors} is "PNPLOG"-hard.
\end{lemma}

Combining Corollary~\ref{cor:restrictor_upper} with Lemma~\ref{lem:restrictors_lower} yields the following tight bound:

\begin{theorem}
	The evaluation problem for "GQL" queries with "restrictors" is "PNPLOG"-complete.
\end{theorem}
This completes the "PNPLOG"-completeness characterization for "GQL" queries with restrictors.

\subsection{Proof of Lemma~\ref{lem:restrictors_lower}: "PNPLOG" Lower Bound}
\begin{figure*}[t]
  \centering
  \resizebox{\linewidth}{!}{\begin{tikzpicture}[
  every label/.style={font=\scriptsize},
  nstyle/.style={circle,draw,inner sep=1pt},
  vstyle/.style={circle,draw,inner sep=1pt,fill=white},
  box/.style={draw,dashed,rounded corners,
              inner xsep=10pt,inner ysep=12pt},
  ->,>=Stealth,shorten >=2pt,shorten <=2pt,node distance=1.5cm
]

\node[nstyle,label=below:$s$] (s0) {};


\node[nstyle,right=1cm of s0,label=below:$s_1$] (s1) {};
\node[nstyle,right=1cm of s1,label=below:$t_1$] (t1) {};
\node[box,fit={(s1) (t1)},label=below:{\small$G_1^{\mathsf{lo}}$}] (G1) {};
\node[vstyle,above=8mm of t1,label=above:{$v_1{:}\ell_0$}] (v1) {};
\draw (t1) -- (v1);

\node[nstyle,right=2.3cm of t1,label=below:$s_2$] (s2) {};
\node[nstyle,right=1cm  of s2,label=below:$t_2$] (t2) {};
\node[box,fit={(s2) (t2)},label=below:{\small$G_2^{\mathsf{lo}}$}] (G2) {};

\node[nstyle,right=2.3cm of t2,label=below:$s_3$] (s3) {};
\node[nstyle,right=1cm  of s3,label=below:$t_3$] (t3) {};
\node[box,fit={(s3) (t3)},label=below:{\small$G_3^{\mathsf{lo}}$}] (G3) {};
\node[vstyle,above=8mm of t3,label=above:{$v_3{:}\ell_0$}] (v3) {};
\draw (t3) -- (v3);

\node[right=1.7cm of t3] (dots1) {$\cdots$};

\node[nstyle,right=1.3cm of dots1,label=below:$s_{2k+1}$] (sk) {};
\node[nstyle,right=1cm  of sk,label=below:$t_{2k+1}$] (tk) {};
\node[box,fit={(sk) (tk)},label=below:{\small$G^{\mathsf{lo}}_{2k+1}$}] (Gk) {};
\node[vstyle,above=8mm of tk,label=above:{$v_{2k+1}{:}\ell_0$}] (vk) {};
\draw (tk) -- (vk);

\node[right=1.7cm of tk] (dots2) {$\cdots$};

\node[nstyle,right=1.3cm of dots2,label=below:$s_n$] (sn) {};
\node[nstyle,right=1cm  of sn,label=below:$t_n$] (tn) {};
\node[box,fit={(sn) (tn)},label=below:{\small$G_n$}] (Gn) {};

\draw (s0) -- (s1);
\draw (t1) -- (s2);
\draw (t2) -- (s3);
\draw (t3) -- (dots1);
\draw (dots1) -- (sk);
\draw (tk) -- (dots2);
\draw (dots2) -- (sn);

\def\rowgap{2.2cm}

\node[nstyle,above=\rowgap of s1,label=above:$s'_2$] (S2u) {};
\node[nstyle,right=1cm of S2u,label=above:$t'_2$] (T2u) {};
\node[box,fit={(S2u) (T2u)},label=above:{\small$G^{\mathsf{up}}_2$}] (G2u) {};

\node[nstyle,above=\rowgap of s2,label=above:$s'_3$] (S3u) {};
\node[nstyle,right=1cm of S3u,label=above:$t'_3$] (T3u) {};
\node[box,fit={(S3u) (T3u)},label=above:{\small$G^{\mathsf{up}}_3$}] (G3u) {};

\node[nstyle,above=\rowgap of s3,label=above:$s'_4$] (S4u) {};
\node[nstyle,right=1cm of S4u,label=above:$t'_4$] (T4u) {};
\node[box,fit={(S4u) (T4u)},label=above:{\small$G^{\mathsf{up}}_4$}] (G4u) {};

\node[right=1.7cm of T4u] (dotsU1) {$\cdots$};

\node[nstyle,above=\rowgap of sk,label=above:$s'_{2k+2}$] (Sk2u) {};
\node[nstyle,right=1cm of Sk2u,label=above:$t'_{2k+2}$] (Tk2u) {};
\node[box,fit={(Sk2u) (Tk2u)},label=above:{\small$G^{\mathsf{up}}_{2k+2}$}] (Gk2u) {};

\node[right=1.7cm of Tk2u] (dotsU2) {$\cdots$};

\draw[densely dotted] (v1) -- (S2u);
\draw[densely dotted] (v1) -- (S3u);
\draw[densely dotted] (v1) -- (S4u);
\draw[densely dotted] (v1) -- (Sk2u);
\draw[densely dotted] (v3) -- (S4u);
\draw[densely dotted] (v3) -- (Sk2u);
\draw[densely dotted] (vk) -- (Sk2u);

\def\topgap{1.4cm} 
\node[nstyle,above=\topgap of dotsU1,label=above:$t$] (TopT) {};

\draw[densely dotted] (T2u) -- (TopT);
\draw[densely dotted] (T3u) -- (TopT);
\draw[densely dotted] (T4u) -- (TopT);
\draw[densely dotted] (Tk2u) -- (TopT);

\end{tikzpicture}}
           \caption{The graph $G^\star$ used in the reduction.  
\textbf{Lower layer:} Copies $G_i^{\mathsf{lo}}$ of each input graph $G_i$ are linked by edges 
$s\to s_1$, $t_i\to s_{i+1}$, and odd‐index nodes $v_j$ (label $\ell_0$) are appended via $t_j\to v_j$.  
\textbf{Upper layer:} Fresh copies $G_j^{\mathsf{up}}$ for $j\ge2$ have sources $s_j'$ connected from every $v_i$ with $i<j$, and all upper targets $t_j'$ feed into the global sink $t$.  
Source nodes $s,s_i,s_j'$ carry label $s$; targets $t,t_i,t_j'$ carry $t$; and each $v_j$ (odd $j$) carries $\ell_0$.}
  \label{fig:gstar}
\end{figure*}

We give a log-space reduction from the "PNPLOG"-complete
{"Odd-Index problem"} 
to the evaluation of a  \textsc{GQL} query that uses the
\gpckw{simple} restrictor.  
(A similar construction works with \gpckw{trail}.)

\smallskip
\noindent
\textbf{Query.} We define the query  $Q \df Q_L \join Q_R$ where 
\[
\begin{aligned}
  Q_L &\coloneqq 
    \gpckw{simple}\,\Bigl(
      ({:}s)\to
  \gpcppat
      \to(x{:}\ell_0)
    \Bigr)_{x} \quad \text{where}\quad 
    \gpcppat  \coloneqq  
      ({:}s)(\to \to)^{*} ({:}t)
     \Big(  \to(:s) (\to \to)^{*} ({:}t)\Big)^*  
    \\
  Q_R &\coloneqq 
    (x{:}\ell_0)_{x}\;
    \setminus\;
    \gpckw{simple}\,\Bigl(
      (x{:}\ell_0)\to(:s) (\to\to)^{*}({:}t)\to({:}t)
    \Bigr)_{x}
\end{aligned}
\]

\smallskip
\noindent
\textbf{Graph construction.}
Assume the input list
\((G_1,s_1,t_1),\dots,(G_n,s_n,t_n)\) is pairwise
disjoint, and each \(s_i,t_i\)  carries labels ${s},{t}$, respectively.
We build \(G^\star\), depicted in Figure~\ref{fig:gstar}, as follows:
\begin{enumerate}
  \item \textit{Lower layer.}
    \begin{enumerate}
      \item For every \(1\le i \le n \) add a fresh copy \(G_i^{\mathrm{lo}}\).
      \item Add a fresh node \(s\) labeled by ${s}$, and an edge \(s\to s_1\).
      \item For each \(i<n\) add an edge \(t_i\to s_{i+1}\).
      \item For each \textit{odd} \(j\) add a node $v_j$ labeled $\ell_0$,
            and an edge \(t_j\to v_j\).
    \end{enumerate}

  \item \textit{Upper layer.}
    \begin{enumerate}
      \item For every \(2\le j\le n\) add a  fresh copy
            \(G_j^{\mathrm{up}}\) with nodes \(s_j',t_j'\) (replacing \(s_j,t_j\)).
      \item Add an edge \(v_j\to s'_{i}\) for every odd \(j\) and $i\ge j+1$.
      \item Add a fresh node  \(t\) labeled by \({t}\), and
            edges \(t_i'\to t\) for all \(i\ge2\).
    \end{enumerate}
\end{enumerate}

Label summary: all sources
\(s,s_1,\dots,s_n,s_2',\dots,s_n'\) carry $s$;
all targets \(t,t_1,\dots,t_n,t_2',\dots,t_n'\) carry $t$;
each \(v_j\) (odd \(j\)) carries \(\ell_0\).

Notice that $G^{\star}$ construction is first-order definable and uses only
logarithmic space.

\smallskip
\noindent
\textbf{Correctness.}
To prove correctness, we investigate both parts $Q_L$ and $Q_R$ of the query in the following lemmas. 
\begin{lemma}
      $Q_L$ evaluated on $G^\star$ returns exactly those nodes $v_j$ such that for every  $i\le j$, the graph $G_i$ contains a simple path of even-length from $s_i$ to $t_i$.
\end{lemma}
\begin{proof}
Assume first that $Q_L$ evaluated on $G^\star$ returns a node $n$.
Due to $Q_L$'s definitiom, 
there must exist in $G^\star$ a \emph{simple} path of the form
\[
  s \;\to\; s_1 \;\overset{p_1}{\leadsto}\; t_1  \;\to\;  s_2 \;\overset{p_2}{\leadsto}\; t_2 
  \;\to\; \cdots  \;\to\;  s_i \;\overset{p_i}{\leadsto}\; t_i  \;\to\;
  n,
\]
where (1) $1\le i \le n$;
(2) \(s\) is the unique fresh node labeled $s$ created in Step 1(b), (3) \(s_i\) is the source of $G_i^{\mathrm{lo}}$ in the the lower layer;
(4) each $p_i$ consists of an even number of edges; and
(5) \(n\) is labeled~\(\ell_0\).

Notice that he only \(\ell_0\)-labelled nodes are
\(v_1,v_3,v_5,\dots\). 
Let us denote \(n=v_j\) for the corresponding odd~\(j\).
By construction \(v_j\) is attached to \(t_j\), hence \(t'=t_j\).

Due to the structure of $G^{\star}$ and to the labels specifies in $Q_L$, the matched path is forced to traverse, in order,
\[
G_1^{\mathrm{lo}},\;G_2^{\mathrm{lo}},\;\dots,\;G_j^{\mathrm{lo}},
\]
employing the connecting edges \(t_i\to s_{i+1}\) for \(i<j\).
As the overall walk is simple, in each lower copy
\(G_i^{\mathrm{lo}}\) the sub-path from \(s_i\) to \(t_i\) is itself
simple.
Consequently, for all $i\le j$, it holds that 
  $G_i$ contains a simple path of even length $s_i\leadsto t_i$ from $s_i$ to $ t_i$.

Conversely, suppose there is an odd index \(j\) such that every
\(G_i\;(i\le j)\) has a simple even-length path \(s_i\leadsto t_i\).
Concatenate those paths with the connecting edges

\[
  s\to s_1,\quad
  t_i\to s_{i+1}\;(i<j),\quad
  t_j\to v_j;
\]

The result is a simple path that matches
$Q_L$, hence \(v_j\) is returned.
\end{proof}
\begin{lemma}
    $Q_R$ evaluated on $G^\star$ returns exactly those nodes $v_j$ such that, for every $i\ge j+1$, the graph $G_i$ contains \emph{no} simple path of even-length from $s_i$ to $t_i$.
\end{lemma}

\begin{proof}

Recall
$
  Q_R \coloneqq 
    (x{:}\ell_0)_{x}\;
    \setminus\;
    \gpckw{simple}\,\Bigl(
      (x{:}\ell_0)\to(:s) (\to\to)^{*}({:}t)\to({:}t)
    \Bigr)_{x}
$
    We refer to the patterns before and after $\setminus$, respectively, as
    \emph{left} and \emph{right} parts. 

\emph{Left part.}
The projection \((x{:}\ell_0)_{x}\) simply returns every
\(\ell_0\)-labelled vertex, i.e.\ the odd-indexed nodes
\(v_1,v_3,v_5,\ldots\).

\smallskip
\emph{Right part.}
A
node \(v_j\) is matched to the right part if there is a path 
\[
  v_j \;\to\; s_i' \;\leadsto\; t_i' \;\to\; t
  \quad\text{for some } i\ge j+1
  \tag{\(*\)}
\]
exists in \(G^\star\).  where  $ s_i' \;\leadsto\; t_i'$ is of even length, 
The edge \(t_i'\to t\) is the \emph{only} occurrence of two consecutive
$t$-labelled nodes, so any matched path must enter the upper
copy \(G_i^{\mathrm{up}}\) at its source \(s_i'\) and reach \(t_i'\)
without repeating nodes.  
Because each \(G_i^{\mathrm{up}}\) is disjoint from the rest of the graph,
such a sub-path exists \emph{iff} the original graph \(G_i\) contains a
simple path of even length \(s_i\leadsto t_i\).

\emph{The difference of left part and right part}
Due to the above, $v_j$ is outputted by $Q_R$ whenever there is no even-length simple path in $G_i^{\mathrm{up}}$ from $s_i$ to $t_i$ for every $i\ge j+1$.
\end{proof}

\textbf{Combining Claims 1 and 2.}
Claim 1 shows that $Q_L$ returns precisely those markers \(v_j\) for which every graph \(G_i\) with \(i\le j\) admits a simple even-length \(s_i\leadsto t_i\) path.  Claim 2 shows that $Q_R$ returns exactly those same markers \(v_j\) for which no graph \(G_i\) with \(i\ge j+1\) admits such a path.  
Since our full query takes the intersection of these two results, it will be nonempty exactly when there exists some index \(j\) satisfying both “all earlier graphs succeed” and “all later graphs fail.”  But this is exactly the odd‐index condition required by the source problem, which completes the correctness argument.

Exactly the same reduction  go through if we replace the \gpckw{simple} restrictor by \gpckw{trail}, yielding identical bound.

\begin{toappendix}
	
\end{toappendix}

	\section{Incorporating Domain-Specific Operators}\label{sec:domspecific}

Logics considered so far let us combine the topology of a graph with simple data tests restricted to equality.
Yet, practical languages operate over concrete domains equipped with a rich set of comparisons and functions.

To add such data operations without worsening complexity, we rely on classical results from the study of constraint databases~\cite{constraint-databases} and embedded finite model theory \cite{fmt-app,collapsejournal}. These results allow us to extend "GQL" with typed data operations while preserving data-complexity bounds.

The underlying idea is to `embed' a finite database inside an infinite structure $\struct$ with a decidable first-order theory.
Queries can then use formulas over $\struct$ to express operations on domain values.
Three notable examples~\cite[Chapter 5]{fmt-app}) include:
\begin{itemize}
    \item The \AP""Linear Integer Arithmetic"" (\reintro{LIA}) structure $\intro*\strLA = \tup{\Z,0, 1, +, \leq}$ — linear arithmetic over integers.
    \item The \AP""Linear Real Arithmetic"" (\reintro{LRA}) structure $\intro*\strRLA = \tup{\reals,0,1, +, \leq}$ — linear arithmetic over reals.

    \item The \AP""Real Ordered Field"" (\reintro{ROF}) structure $\intro*\strROF = \tup{\reals,0,1, +, \cdot, \leq}$ —  real closed field (Tarski’s structure).
\end{itemize}

We focus on these numeric structures, though similar approaches exist for other domains like strings (e.g., reducts of the universal automatic structure~\cite{BLSS03}).

	\begin{toappendix}

\subsection{Embedded Finite Models: The Setting}
We briefly introduce here the setting of embedded finite models 
\cite[Chapter 5]{fmt-app}. 

\partitle{The model} Consider an \AP""infinite structure"" $\istruc$ with universe $\domain$, over a \AP""vocabulary"" $\vocab_\istruc$ containing both ""function names"" 
$f_1,\ldots,f_n$ 
and "relation names" $R_1,\ldots,R_m$.
We shall henceforth blur the distinction between function/relation names and their interpretation 
in $\istruc$.

Fix a "relational vocabulary" $\vocab$, disjoint from the "vocabulary" of 
$\istruc$. 
\sidediego{imported defn from body}
\AP	An element $u \in D$ is  ""definable"" over $\istruc$ 
if there exists a $\istruc$-formula $\varphi(x)$ such that for each $v \in D$ we have $\sem{\varphi(x)}{x\mapsto u}{\istruc} = \top$ iff $u=v$.
For example, the set of "definable" elements over $\strRLA$ is precisely the set of rational numbers.

A ""$\istruc$-embedded finite relational structure"" over $\vocab$ is a finite "relational $\vocab$-structure" $\struct$ such that $\Const$ is taken to be 
the set of first-order "definable" 
elements over $\istruc$.  
That is, each "relation name" of $\vocab$ is interpreted in $\struct$ as a finite relation over $\istruc$-"definable" elements.

\partitle{The logic} We define progressively more expressive logics over  "$\istruc$-embedded structures": 
(1)
first-order logic $\FO[\istruc]$, (2) hybrid transitive closure logic
$\FOHTC[\istruc]$, (3) hybrid second-order logic $\HSO[\istruc]$. 

\AP
A formula $\phi$ in $\FO[\istruc]$ over the "relational vocabulary" $\vocab$ is simply a 
first-order formula over the "vocabulary" $\vocab' := \vocab \dcup 
\vocab_\istruc$, where we additionally allow the ""active-domain quantifier"" 
$\existsRadom x$.
\AP
Here, the quantifier $\intro*\existsRadom x$ asserts the existence of an element in
the "active domain" of a given "relational structure" $\struct$, whereas the standard 
quantifier $\exists x$ asserts 
the existence of an element in the "domain" $\domain$ of $\istruc$ ("ie", not
necessarily in the "active domain"). The meaning of the formula is immediate from
first-order logic, i.e., we interpret this over the structure $\istruc'$ over
the vocabulary $\vocab'$ with domain $\domain$ and relations/functions from
$\istruc$ and $\struct$. 
Such a formula $\phi$ is said to be an ""active-domain formula"", if every quantifier therein is an "active-domain quantifier". We now restate a classical result from embedded finite model theory (see, "eg",
\cite{collapsejournal,fmt-app}),
which is often called \AP""Restricted Quantifier Collapse"" (\reintro{RQC}). The intuition behind "RQC" is that the "collapse@Restricted Quantifier Collapse" lets us restrict the range of a quantifier to the active domain, hence reducing evaluation complexity.
\begin{proposition}[\cite{collapsejournal}]
    Let $\istruc$ be either "LRA", "LIA", or "ROF".
    Each $\FO[\istruc]$-formula 
    is equivalent to an
    "active-domain $\FO[\istruc]$-formula". Therefore, the data complexity for
    $\FO[\istruc]$ is in $\TCzero$.
    \label{prop:FO-embedded}
\end{proposition}
The proposition results in $\TCzero$ complexity because an "active-domain formula" can be evaluated like a normal $\FO$ formula with a key difference: Each atomic formula from $\istruc$ is evaluated using decision procedures from $\istruc$. This leads to $\TCzero$ instead of $\ACzero$ (the data complexity of $\FO$), as each fixed $\istruc$-formula can be evaluated in $\TCzero$. Specifically, for a fixed $\istruc$-formula $\varphi(\bar{x})$, the problem of deciding whether $\sem{\varphi(\bar x)}{\mu}{\istruc}$ holds for a given interpretation $\mu$ of the variables $\bar{x}$ is in $\TCzero$, particularly when the formula is quantifier-free (as in "LRA", "LIA", and "ROF"), which admit quantifier elimination.

Interestingly, as shown in~\cite{collapsejournal}, the 
above proposition has been extended to more powerful logics
including $\FOTC$ and $\SO$. The trick is to interpret the ``higher-order
features'' (i.e. transitive closure operators and relation quantifications)
with an active-domain interpretation. The resulting logics are called
\AP""Hybrid Transitive Closure Logic"" ($\intro*\FOHTC$) and "Hybrid Second-Order Logic" ($\HSO$),
respectively. More precisely, syntactically, we define $\FOHTC[\istruc]$ as an 
extension of $\FO[\istruc]$ with the $\TCemb$-operators as defined for $\FOTC$
in Section \ref{sec:prelim}.  Semantically, for any interpretation
$\nu$ in a "$\istruc$-embedded structure" $\struct$, the meaning of
\AP
\begin{align*}
[\intro*\TCemb_{\bar u,\bar v}(\phi(\bar u \bar v \bar p))](\bar x, \bar y) 
\end{align*}
is that there exists a sequence of $k$-tuples of "active domain" elements 
$\bar d_0, \dotsc, \bar d_n \in \adom(\struct)^k$ such that 
    (i) $\sem{\phi}{\bar u \bar v \bar p \mapsto \bar d_i \bar d_{i+1} \nu(\bar p)}{\struct} = \top$ for every $i$, 
    (ii) $\bar d_0 = \nu(\bar x)$, and 
    (iii) $\bar d_n = \nu(\bar y)$.
The definition of an "active-domain formula" from $\FO[\istruc]$ extends to
$\FOHTC[\istruc]$, that is, each quantifier must be an "active-domain quantifier".
\begin{proposition}[\cite{collapsejournal}]
    Let $\istruc$ be either "LRA", "LIA", or "ROF".
    Each $\FOHTC[\istruc]$-formula is equivalent to an
    "active-domain $\FOHTC[\istruc]$-formula". Thus, the data complexity for
    $\FOHTC[\istruc]$ is in "NL".
    \label{prop:HTC}
\end{proposition}
The reasoning for the data complexity in this proposition is the same as for
Proposition \ref{prop:FO-embedded}. Here, we obtain "NL" because the data
complexity for $\FOTC$ is in "NL" and  $\TCzero \subseteq \NL$.

\AP
Finally, we define $\intro*\HSO[\istruc]$ (""Hybrid Second-Order Logic"") syntactically as an extension of 
$\FO[\istruc]$ with
second-order quantifiers as defined for $\SO$ in Section \ref{sec:prelim}.
For the second-order quantifiers, we will write $\existsRadom R$ instead of 
$\exists R$. This is because we will interpret $\exists R\varphi$ as there
exists a relation over $\adom(\struct)$ such that $\varphi$ holds. We define
$\intro*\HESO[\istruc]$ as a restriction of $\HSO[\istruc]$ to formulas of the form
$
    \existsRadom R_1,\ldots,R_n ~ \varphi
$
where $\varphi$ is an $\FO[\istruc]$-formula over $\vocab_\istruc \cup \vocab
\cup \{R_1,\ldots,R_n\}$.
\begin{proposition}[\cite{collapsejournal}]
    Let $\istruc$ be either "LRA", "LIA", or "ROF".
    Each $\HSO[\istruc]$-formula (resp. $\HESO[\istruc]$-formula) 
    is equivalent to an
    "active-domain $\HSO[\istruc]$-formula" (resp. "active-domain $\HESO[\istruc]$-formula"). 
    Therefore, the data complexity for $\HSO[\istruc]$ (resp. $\HESO[\istruc]$)
    is in "PH" (resp. "NP").
    \label{prop:HSO}
\end{proposition}
The reasoning for the data complexity in this proposition is the same as for
Proposition \ref{prop:FO-embedded}. Here, we obtain PH (resp.~"NP") because the 
data complexity for $\SO$ (resp.~$\ESO$) is in PH (resp.~"NP") and that 
$\TCzero$ $\subseteq$ "NP" $\subseteq$ "PH".
\end{toappendix}

\subsection{Embedded Finite Model Framework}\label{sec:efm}

\partitle{Embedded finite structures}
\AP	An element $u$ from the domain $D$ of $\istruc$ is  ""definable"" over $\istruc$ 
if there exists a $\istruc$-formula $\varphi(x)$ such that for each $v \in D$ we have $\sem{\varphi(x)}{x\mapsto u}{\istruc} = \top$ iff $u=v$.
	For example, over $\strRLA$, the "definable" elements are exactly the rationals. 
    An ""$\istruc$-embedded finite relational structure"" over $\vocab$ 
    is a finite relational structure whose domain consists of such definable elements.

\partitle{Logics with active–domain quantifiers}
A hybrid transitive-closure operator
\[
  \TCemb_{\bar u,\bar v}\bigl(\varphi(\bar u\,\bar v\,\bar p)\bigr)
\]
is defined like the standard $\TC$ except that every intermediate tuple  
must lie in the active domain.  We denote by \AP$\intro*\FOHTC$ first-order logic extended with hybrid TC operators.


Active-domain restrictions also apply to second-order quantifiers. In \AP$\intro*\HSO$ for hybrid second-order logic, relation variables range only over the active domain. We write
\AP$\intro*\HESO$ for its existential fragment.

For any $L\in\{\FOHTC,\HSO,\HESO\}$, let
$L[\istruc]$ denote the logic extended with predicates from an infinite structure $\istruc$,
enriched with \emph{active-domain quantifiers}
$\exists^{\mathrm{adom}}x$ and $\exists^{\mathrm{adom}}R$.
A formula is \emph{active-domain} if all quantifiers are of this form.

\partitle{Restricted-quantifier collapses}
Over the arithmetic domains introduced above, quantifiers can be restricted to the active domain without losing expressive power:

\begin{proposition}[Restricted-quantifier collapse {\cite{collapsejournal}}]
\label{prop:rqc}
Let $\istruc\in\{\strLA,\strRLA,\strROF\}$.  
Over "$\istruc$-embedded structures":
\begin{enumerate}[(i)]
  \item Every $\FOHTC[\istruc]$-formula is equivalent to an \emph{active-domain} $\FOHTC[\istruc]$-formula, and its data complexity lies in "NL".
  \item Every $\HESO[\istruc]$-formula is equivalent to an \emph{active-domain} $\HESO[\istruc]$-formula, and its data complexity lies in "NP".
\end{enumerate}
\end{proposition}

These collapses allow us enriching "GQL" with
operations from $\istruc$ while keeping the
usual bounds.

	\subsection{Extending "GQL" with Data Types}

\OMIT{
\paragraph{Extending Conditions}
In GPC querying the actual data (that is, labels, keys and properties) is within  the conditions $\theta$ in patterns of the form $\gpcppat_{\langle \theta \rangle }$. Atomic $\theta$ is either $x:\ell$ or $x.a = y.b$ or $x.a = c$.  
Enriching conditions with domain-specific operators and predicates can be quite useful as the following example shows.
\begin{example}
    Assume we are interested in detecting a money laundering scheme. To this end, we may be interested in tracking an account $x$ with a cyclic sequence of bank transfers with increasing time stamps, starting and ending in $(x)$.
    Assume we have a property graph $G$ whose nodes are bank accounts labeled with `suspicious' or `not suspicious', and edges are bank transfers with time-stamp attached to the key $c^{\mathsf{ts}}$.
    \begin{align*}
        [\TC_{u,v}(()\overset{u}{\rightarrow}()\overset{v}{\rightarrow}()_{\langle u.c^{\mathsf{ts}} >v.c^{\mathsf{ts}}\rangle})](x,x)
    \end{align*}
\end{example}
}
\AP
We fix an infinite structure $\istruc$. 
and define "GQL" queries with an $\istruc$-data type.
To this end, we first extend the definition of "terms" to $\istruc$-terms defined as follows:
\[
\gpcterm \df x.a \mid c\mid y   \mid f^m(\gpcterm_1,\ldots , \gpcterm_m) 
\]
where $c$ is an $\istruc$-"definable" constant, $y$ is a variable ranging 
over the domain of $\istruc$, $f^m$ is an $m$-ary "function name" in the "vocabulary" of $\istruc$, 
$x \in \Vars$, $a\in\KeySet$, and every $\gpcterm_i $ is an $\istruc$-"term".
We also extend the definition of "conditions" to $\istruc$-"conditions" by adding
$
\theta \df  R^m(\gpcterm_1,\ldots, \gpcterm_m) 
$
where $R^m$ is an $m$-ary "relation name" in the "vocabulary" of $\istruc$, and every $\gpcterm_i $ is an $\istruc$-"term". 
The semantics of $\theta$ extends by interpreting $\gpcterm$ and $\theta$
``inside'' $\istruc$.


\OMIT{
naturally where a constant (resp. variable
$y$) 
uninterpreted constants may 
be assigned to every value from the infinite domain (particularly, a value that does not appear in the input graph). 
}

\OMIT{
\begin{example}\sidediego{This goes in the embedded models section}
    Considering our running \Cref{ex:example-bikelanes}, suppose now that we have a relation $\textsc{Flooded}(\textit{x-coord},\textit{y-coord},\textit{radius})$ with reports of flooded areas due to recent storms. Is there a way to go from Meerkerk to Asperen by Bike lanes avoiding flooded areas? We can express such a property assuming a modelling on "ROF" with the following $\FOHTC$ sentence:
    \begin{align*}
        \exists z,z' ~ \TC_{u,v}\Big(&\exists e ~ E_u(e) \land \text{src}(e,u) \land \text{tgt}(e,v) \land \text{label}(e,\text{Bike-lane}) \land {}\\&
        \big(\exists x,y ~ \text{property}(u,\text{`x-coord'},x) \land \text{property}(u,\text{`y-coord'},y) \land {}\\&
        (\forall x',y',r ~
        \text{Flooded}(x',y',r) \rightarrow 
        (x-x')^2 + (y-y')^2 > r^2)\big)\Big)[z,z'] \land {}\\&
        \text{property}(z,\text{`name'},\text{`Meerkerk'}) \land
        \text{property}(z',\text{`name'},\text{`Asperen'})
    \end{align*}
\end{example}
}

\begin{example}
Continuing~\Cref{ex:fotc}, suppose we also want to enforce that the Manhattan distance between the towns is below a fixed \(c\).  We can add the filter
\[
|x.\text{x{-}coord} - y.\text{x{-}coord}| \;+\; |x.\text{y{-}coord} - y.\text{y{-}coord}| \;<\; c,
\]
which requires access to built-in arithmetic relations for absolute value, addition, subtraction, and comparison $<$.  
\end{example}

\OMIT{
\begin{example}\sidediego{This goes in the embedded models section}
    Suppose that in \Cref{ex:socialnet} we want to find out the maximum ages on each friendship path between Yan and Vim..
    \begin{align*}
        \phi(x_{max})=\exists z,z',\hat z ~ &\TC_{u,v}\Big(\exists e ~ E_u(e) \land \text{src}(e,u) \land \text{tgt}(e,v) \land \text{label}(e,\text{Friends}) 
        \land {}\\&\hspace{0,95cm}\exists x_{age} ~ \text{property}(v,\text{`age'},x_{age}) \land 
        x_{age} \leq x_{max} \Big)[z,\hat z] \land {}\\&
        \TC_{u,v}\Big(\exists e ~ E_u(e) \land \text{src}(e,u) \land \text{tgt}(e,v) \land \text{label}(e,\text{Friends}) 
        \land {}\\&\hspace{0,95cm}\exists x_{age} ~ \text{property}(v,\text{`age'},x_{age}) \land 
        x_{age} \leq x_{max} \Big)[\hat z,z'] \land {}\\&
        \exists x_{name} ~ \text{property}(z,\text{`name'},x_{name}) \land x_{name} = \text{`Yan'}\land {}\\&
        \exists x_{name} ~ \text{property}(z',\text{`name'},x_{name}) \land x_{name} = \text{`Vim'}\land {}\\&
        \exists x_{age} ~ \text{property}(\hat z,\text{`age'},x_{age}) \land x_{age} = x_{max}
    \end{align*}
\end{example}
}

\OMIT{
Notice that we can extend "GQL" with more than one theory by simply setting $\istruc$ as the union of all structures we are interested in.

We can extend our upper bounds to these settings, for those theories that enjoy the "RQC" property.

\sideliat{should be probably moved somewhere}
}

Using~\Cref{prop:rqc}, the upper bounds extend to GQL queries with data types.

\begin{theorem}
  For every $\istruc \in \set{\strLA,\strRLA,\strROF}$, the "data complexity" of
    the "evaluation problem@@GQL" of a Boolean "GQL" query $\gpcq$ with an 
    $\istruc$-data type is
  \begin{itemize}
      \item "NL"-complete if $\gpcq$ is 
      "restrictor-free";
      \item "PNPLOG"-complete otherwise.
  \end{itemize}
\end{theorem}

\OMIT{
\begin{remark}
    It is easy to extend "GQL" to multiple data types (e.g. strings
    and numbers) by employing multiple infinite structures having
    "Restricted Quantifier Collapse" and typed variables,
    while preserving the same data complexity.
    The required condition here is that 
    variables of different types are not allowed to ``communicate'' with each 
    other (i.e. each theory has disjoint function and relation symbols).
\end{remark}}
"GQL" can support multiple data types (e.g., strings, numbers) using disjoint infinite structures with "restricted quantifier collapse", while preserving data complexity as long as typed variables remain isolated.

\OMIT{
\begin{example}
  Continuing our bike-lane example, assume we incorporate a theory of strings enjoying the "RQC" property and supporting the predicate
$
\mathit{prefix}(s, t)
$
that holds exactly when the string \(s\) is a prefix of \(t\). We refine our reachability check by conjoining our filterng condition with
$$
\exists s\;\bigl(\mathit{prefix}(s,\,x.\mathit{name})\;\land\;\mathit{prefix}(s,\,y.\mathit{name})\bigr)
$$
where \(s\) ranges over strings. This ensures that, in addition to being connected by bike lanes, the town names share a common prefix.
\end{example}
}

	\section{Logical Graph Querying Benefits}\label{sec:overview}

This section advocates viewing graph query languages through a logical lens. Embedding "GQL" into formal logic simplifies reasoning, clarifies expressiveness, and preserves complexity bounds. We showcase several examples unlocked by this embedding, including schema-level queries and extensions beyond plain "GQL".

\partitle{Meta-Querying} "GQL" does not support quantification over properties or labels, limiting its ability to express meta-level queries such as schema checks or property comparisons. For instance, to return pairs of nodes that hold exactly the same data, we have:
    \[
    \phi_{eq}(n,n') = \forall k, v  ~ \left(\propf(n,k,v) \leftrightarrow \propf(n',k,v)\right)
    \]
 where $\leftrightarrow$ is used as a shorthand for bidirectional implication.
Similarly, to ensure that all bike lanes in~\Cref{ex:fotc} are toll-free, we can replace $\phi(u,v)$ with 
\[
       \phi(u,v) \df \exists e\, \Big(\ E(e) \land \srcf(e,u) \land \tgtf(e,v) \land \labelf(e,\text{``BikeLane''})
       \land \mathsf{no\_toll}(e)
       \Big)
\]
where  $\mathsf{no\_toll}(e) \df\neg \exists x: \propf(e,\text{`toll'},x)$ verifies the condition holds.
While "GQL" cannot test for the absence of a toll property, this enriched query remains in "NL".

\partitle{Beyond GQL's Expressiveness} 
"GQL" is limited in checking conditions on unbounded edge traversals, for instance, checking for decreasing values on edges along a path~\cite{Gheerbrant2025GQL}. The next query expresses this: outputs towns accessible by bike lanes with decreasing lengths:  
\[\phi(n,n') \df \exists e,e'\,\left(\phi_{<}(e,e') \wedge \tgtf({e,n}) \wedge \srcf({e',n})
\right)\] where 
\begin{equation*}
\phi_{<}(e,e') \df[ \TC_{u,v}
\Big(
{E}(u) \wedge {E}(v) \wedge \phi_{\text{incE}}(u,v)  
 \Big)
 ](e,e')
\end{equation*}
where $\phi_{\text{incE}}(u,v)$ checks if $u$ and $v$ are adjacent edges with decreasing lengths:
\begin{multline*}
\phi_{\text{incE}}(u,v) \df \exists n\, (\srcf({v,n})   \wedge  \tgtf({u,n})) \wedge   \exists l,l'\\(\propf(u,\text{``length''},l) \wedge \propf(v,\text{``length''},l') \wedge l<l')
\end{multline*}
Here we compute the transitive closure over edges rather than nodes, as it is in "GQL"'s iteration.

	\section{Conclusion}\label{sec:conc}
	
In this paper, we have answered arguably one of the most fundamental questions 
regarding evaluating "GQL" queries, i.e., its data complexity. Although its
complexity ("PNPLOG") is higher than that of its SQL counterpart (in 
polynomial-time), we have shown that the restrictor-free fragment is in "NL".
We have shown that our results extend to the setting with data.

\partitle{Related Work}
The study of graph databases in database theory can be traced back to the work
of Mendelzon and Wood~\cite{MW95} on Regular Path Queries (RPQ), where
they popularized the data model of \emph{edge-labeled graphs}. This has been the
standard graph data model for a long time in database theory
\cite{bonifati-book,CGLV00,MW95,BLLW12,FLS98,HKVZ13}). Since then 
the area has undergone several paradigm shifts on the notion of a ``graph data 
model'', including \emph{data graphs} \cite{LMV16} and \emph{property graphs}
\cite{pablo-survey,gql,gpc-pods}. Data graphs were introduced in
\cite{LV12,LMV16} 
owing to the lack
of a treatment of ``data'' (properties associated with nodes and edges, like
``id'', ``name'' and ``age'' of a person) in the data model of edge-labeled graphs. Many classical query languages (like RPQ) have been extended to also
incorporate such data \cite{pablo-survey,LMV16,BFL15,DBLP:conf/pods/FigueiraJL22}, e.g., in the query 
language RDPQ (Regular Data Path Queries). Data graphs still do
not capture the graph model that has been adopted by practitioners, e.g., the
``schemaless'' nature of a graph. The property graph data model
\cite{pablo-survey,gql} has been 
proposed to address such deficiencies.

Unlike the case of GQL, the precise data complexity of most query languages for edge-labeled graphs (e.g. RPQs and extensions) and data graphs (e.g. RDPQs) are known, which are "NL"-complete. For example, see \cite{MW95,BLLW12,FLS98,HKVZ13,bonifati-book,LMV16}. These results were also recently extended to the case of numeric and string data domains and operations \cite{DBLP:conf/pods/FigueiraJL22}. The reason for the low computational complexity in this case is similar to the case of "GQL" without restrictors, which lead to "NP"-hardness \cite{MartensNP23}. Incidentally, our technique of relational embedding into can be shown to provide a simpler proof of the "NL" upper bound in \cite{DBLP:conf/pods/FigueiraJL22}.

\partitle{Future Work}
We believe that our logical perspectives on graph databases could further advance
the study of graph databases. As we have shown, not only do they generalize "GQL" in expressivity (e.g. meta-querying, cf. Section \ref{sec:overview}), they also allow us to employ classical techniques from relational structures (e.g. embedded finite model theory). 

One interesting future avenue
concerns the extension of GQL with aggregation over paths (by, \eg, summing, counting, averaging, 
over a given path). While aggregation results are available for relational databases (\eg~in embedded finite model theory), most positive results (\eg~
see Chapter 5 of \cite{fmt-app}) do not extend to paths and remain within
first-order logic. 

    \paragraph{Acknowledgment.} We thank Michael Benedikt and Leonid Libkin
    for the valuable feedback.
    Lin is supported by the European Research 
    Council\footnote{https://doi.org/10.13039/100010663} under Grant 
    No.~101089343 (LASD). 
	
	\bibliographystyle{plainurl}
	\bibliography{references}

\begin{thebibliography}{10}

\bibitem{AbiteboulHullVianu1995}
Serge Abiteboul, Richard Hull, and Victor Vianu.
\newblock {\em Foundations of Databases}.
\newblock Addison--Wesley, Reading, MA, 1995.

\bibitem{pablo-survey}
Renzo Angles, Marcelo Arenas, Pablo Barcel{\'{o}}, Aidan Hogan, Juan~L.
  Reutter, and Domagoj Vrgoc.
\newblock Foundations of modern query languages for graph databases.
\newblock {\em {ACM} Comput. Surv.}, 50(5):68:1--68:40, 2017.
\newblock \href {https://doi.org/10.1145/3104031} {\path{doi:10.1145/3104031}}.

\bibitem{BFL15}
Pablo Barcel{\'{o}}, Ga{\"{e}}lle Fontaine, and Anthony~Widjaja Lin.
\newblock Expressive path queries on graph with data.
\newblock {\em Log. Methods Comput. Sci.}, 11(4), 2015.
\newblock \href {https://doi.org/10.2168/LMCS-11(4:1)2015}
  {\path{doi:10.2168/LMCS-11(4:1)2015}}.

\bibitem{BLLW12}
Pablo Barcel{\'{o}}, Leonid Libkin, Anthony~Widjaja Lin, and Peter~T. Wood.
\newblock Expressive languages for path queries over graph-structured data.
\newblock {\em {ACM} Trans. Database Syst.}, 37(4):31:1--31:46, 2012.
\newblock \href {https://doi.org/10.1145/2389241.2389250}
  {\path{doi:10.1145/2389241.2389250}}.

\bibitem{benediktsurvey}
Michael Benedikt.
\newblock Generalizing finite model theory.
\newblock In {\em Logic Colloquium '03}, page 3–24. Cambridge University
  Press, 2006.

\bibitem{collapsejournal}
Michael Benedikt and Leonid Libkin.
\newblock Relational queries over interpreted structures.
\newblock {\em J. ACM}, 47(4):644–680, 2000.
\newblock \href {https://doi.org/10.1145/347476.347477}
  {\path{doi:10.1145/347476.347477}}.

\bibitem{BLSS03}
Michael Benedikt, Leonid Libkin, Thomas Schwentick, and Luc Segoufin.
\newblock Definable relations and first-order query languages over strings.
\newblock {\em J. {ACM}}, 50(5):694--751, 2003.
\newblock \href {https://doi.org/10.1145/876638.876642}
  {\path{doi:10.1145/876638.876642}}.

\bibitem{CGLV00}
Diego Calvanese, Giuseppe~De Giacomo, Maurizio Lenzerini, and Moshe~Y. Vardi.
\newblock Containment of conjunctive regular path queries with inverse.
\newblock In {\em {KR}}, 2000.

\bibitem{gql}
Alin Deutsch, Nadime Francis, Alastair Green, Keith Hare, Bei Li, Leonid
  Libkin, Tobias Lindaaker, Victor Marsault, Wim Martens, Jan Michels, Filip
  Murlak, Stefan Plantikow, Petra Selmer, Oskar van Rest, Hannes Voigt, Domagoj
  Vrgoc, Mingxi Wu, and Fred Zemke.
\newblock Graph pattern matching in {GQL} and {SQL/PGQ}.
\newblock In {\em {SIGMOD}}, 2022.
\newblock \href {https://doi.org/10.1145/3514221.3526057}
  {\path{doi:10.1145/3514221.3526057}}.

\bibitem{DBLP:conf/sigmod/DeutschFGHLLLMM22}
Alin Deutsch, Nadime Francis, Alastair Green, Keith~W. Hare, Bei Li, Leonid
  Libkin, Tobias Lindaaker, Victor Marsault, Wim Martens, Jan Michels, Filip
  Murlak, Stefan Plantikow, Petra Selmer, Oskar van Rest, Hannes Voigt, Domagoj
  Vrgoc, Mingxi Wu, and Fred Zemke.
\newblock Graph pattern matching in {GQL} and {SQL/PGQ}.
\newblock In Zachary~G. Ives, Angela Bonifati, and Amr~El Abbadi, editors, {\em
  {SIGMOD} '22: International Conference on Management of Data, Philadelphia,
  PA, USA, June 12 - 17, 2022}, pages 2246--2258. {ACM}, 2022.
\newblock \href {https://doi.org/10.1145/3514221.3526057}
  {\path{doi:10.1145/3514221.3526057}}.

\bibitem{DBLP:conf/pods/FigueiraJL22}
Diego Figueira, Artur Jeż, and Anthony~W. Lin.
\newblock Data path queries over embedded graph databases.
\newblock In {\em {PODS}}, 2022.
\newblock \href {https://doi.org/10.1145/3517804.3524159}
  {\path{doi:10.1145/3517804.3524159}}.

\bibitem{FLS98}
Daniela Florescu, Alon~Y. Levy, and Dan Suciu.
\newblock Query containment for conjunctive queries with regular expressions.
\newblock In {\em PODS}, 1998.
\newblock \href {https://doi.org/10.1145/275487.275503}
  {\path{doi:10.1145/275487.275503}}.

\bibitem{gpc-pods}
Nadime Francis, Am{\'{e}}lie Gheerbrant, Paolo Guagliardo, Leonid Libkin,
  Victor Marsault, Wim Martens, Filip Murlak, Liat Peterfreund, Alexandra
  Rogova, and Domagoj Vrgoc.
\newblock {GPC:} {A} pattern calculus for property graphs.
\newblock In {\em PODS}, 2023.

\bibitem{ICDT23}
Nadime Francis, Am{\'{e}}lie Gheerbrant, Paolo Guagliardo, Leonid Libkin,
  Victor Marsault, Wim Martens, Filip Murlak, Liat Peterfreund, Alexandra
  Rogova, and Domagoj Vrgoc.
\newblock A researcher's digest of {GQL} (invited talk).
\newblock In {\em {ICDT}}, volume 255 of {\em LIPIcs}, pages 1:1--1:22. Schloss
  Dagstuhl - Leibniz-Zentrum f{\"{u}}r Informatik, 2023.

\bibitem{Gheerbrant2025GQL}
Amelie Gheerbrant, Leonid Libkin, Liat Peterfreund, and Alexandra Rogova.
\newblock {GQL and SQL/PGQ: Theoretical Models and Expressive Power}.
\newblock {\em Proceedings of the VLDB Endowment (PVLDB)}, 18(6):1798--1810,
  2025.
\newblock \href {https://doi.org/10.14778/3725688.3725707}
  {\path{doi:10.14778/3725688.3725707}}.

\bibitem{GMT09}
Stefan G{\"{o}}ller, Richard Mayr, and Anthony~Widjaja To.
\newblock On the computational complexity of verifying one-counter processes.
\newblock In {\em Proceedings of the 24th Annual {IEEE} Symposium on Logic in
  Computer Science, {LICS} 2009, 11-14 August 2009, Los Angeles, CA, {USA}},
  pages 235--244. {IEEE} Computer Society, 2009.
\newblock \href {https://doi.org/10.1109/LICS.2009.37}
  {\path{doi:10.1109/LICS.2009.37}}.

\bibitem{Gottlob95}
Georg Gottlob.
\newblock Np trees and carnap's modal logic.
\newblock {\em J. ACM}, 42(2):421–457, March 1995.
\newblock \href {https://doi.org/10.1145/201019.201031}
  {\path{doi:10.1145/201019.201031}}.

\bibitem{GQLStandards}
{GQL Standards Committee}.
\newblock Gql standards website, 2024.
\newblock Accessed: November 2024.
\newblock URL: \url{https://www.gqlstandards.org/}.

\bibitem{fmt-app}
E.~Gr\"{a}del, P.~G. Kolaitis, L.~Libkin, M.~Marx, J.~Spencer, M.~Y. Vardi,
  Y.~Venema, and S.~Weinstein.
\newblock {\em Finite Model Theory and Its Applications}.
\newblock Springer, 2007.

\bibitem{LDBC:TR:TR-2021-01}
Alastair Green, Paolo Guagliardo, and Leonid Libkin.
\newblock Property graphs and paths in gql: Mathematical definitions.
\newblock Technical Reports TR-2021-01, Linked Data Benchmark Council (LDBC),
  Oct 2021.
\newblock URL:
  \url{https://ldbcouncil.org/docs/papers/LDBC-Technical-Report-TR-2021-01--Property-graphs-and-paths-in-GQL-Mathematical-definitions.DOI.10.54285_ldbc.TZJP7279.pdf},
  \href {https://doi.org/10.54285/ldbc.TZJP7279}
  {\path{doi:10.54285/ldbc.TZJP7279}}.

\bibitem{HKVZ13}
Jelle Hellings, Bart Kuijpers, Jan~Van den Bussche, and Xiaowang Zhang.
\newblock Walk logic as a framework for path query languages on graph
  databases.
\newblock In {\em {ICDT}}, 2013.
\newblock \href {https://doi.org/10.1145/2448496.2448512}
  {\path{doi:10.1145/2448496.2448512}}.

\bibitem{bonifati-book}
H.~V. Jagadish, Angela Bonifati, George Fletcher, Hannes Voigt, and Nikolay
  Yakovets.
\newblock {\em Querying Graphs}.
\newblock Morgan and Claypool Publishers, 2018.

\bibitem{constraint-databases}
Gabriel Kuper, Leonid Libkin, and Jan Paredaens.
\newblock {\em Constraint Databases}.
\newblock Springer, 2000.

\bibitem{libkinbook}
Leonid Libkin.
\newblock {\em Elements of finite model theory}, volume~41.
\newblock Springer, 2004.

\bibitem{LMV16}
Leonid Libkin, Wim Martens, and Domagoj Vrgoc.
\newblock Querying graphs with data.
\newblock {\em J. {ACM}}, 63(2):14:1--14:53, 2016.
\newblock \href {https://doi.org/10.1145/2850413} {\path{doi:10.1145/2850413}}.

\bibitem{LV12}
Leonid Libkin and Domagoj Vrgoc.
\newblock Regular path queries on graphs with data.
\newblock In Alin Deutsch, editor, {\em 15th International Conference on
  Database Theory, {ICDT} '12, Berlin, Germany, March 26-29, 2012}, pages
  74--85. {ACM}, 2012.
\newblock \href {https://doi.org/10.1145/2274576.2274585}
  {\path{doi:10.1145/2274576.2274585}}.

\bibitem{MartensNP23}
Wim Martens, Matthias Niewerth, and Tina Popp.
\newblock A trichotomy for regular trail queries.
\newblock {\em Log. Methods Comput. Sci.}, 19(4), 2023.
\newblock URL: \url{https://doi.org/10.46298/lmcs-19(4:20)2023}, \href
  {https://doi.org/10.46298/LMCS-19(4:20)2023}
  {\path{doi:10.46298/LMCS-19(4:20)2023}}.

\bibitem{MW95}
Alberto~O. Mendelzon and Peter~T. Wood.
\newblock Finding regular simple paths in graph databases.
\newblock {\em {SIAM} J. Comput.}, 24(6):1235--1258, 1995.
\newblock \href {https://doi.org/10.1137/S009753979122370X}
  {\path{doi:10.1137/S009753979122370X}}.

\bibitem{SV00}
Holger Spakowski and J{\"{o}}rg Vogel.
\newblock Theta\({}_{\mbox{2}}\)\({}^{\mbox{p}}\)-completeness: {A} classical
  approach for new results.
\newblock In Sanjiv Kapoor and Sanjiva Prasad, editors, {\em Foundations of
  Software Technology and Theoretical Computer Science, 20th Conference, {FST}
  {TCS} 2000 New Delhi, India, December 13-15, 2000, Proceedings}, volume 1974
  of {\em Lecture Notes in Computer Science}, pages 348--360. Springer, 2000.
\newblock \href {https://doi.org/10.1007/3-540-44450-5\_28}
  {\path{doi:10.1007/3-540-44450-5\_28}}.

\bibitem{Vardi82}
Moshe~Y. Vardi.
\newblock The complexity of relational query languages (extended abstract).
\newblock In {\em Proceedings of the Fourteenth Annual ACM Symposium on Theory
  of Computing}, STOC '82, page 137–146, New York, NY, USA, 1982. Association
  for Computing Machinery.
\newblock \href {https://doi.org/10.1145/800070.802186}
  {\path{doi:10.1145/800070.802186}}.

\bibitem{Wagner87}
Klaus~W. Wagner.
\newblock More complicated questions about maxima and minima, and some closures
  of {NP}.
\newblock {\em Theor. Comput. Sci.}, 51:53--80, 1987.
\newblock \href {https://doi.org/10.1016/0304-3975(87)90049-1}
  {\path{doi:10.1016/0304-3975(87)90049-1}}.

\bibitem{wagner1987more}
Klaus~W Wagner.
\newblock More complicated questions about maxima and minima, and some closures
  of np.
\newblock {\em Theoretical Computer Science}, 51(1-2):53--80, 1987.

\end{thebibliography}
	
\end{document}